\documentclass[11pt]{article}
\usepackage[utf8]{inputenc}
\usepackage{style}

\title{Sampling from the Hardcore Model on Random\\ Regular Bipartite Graphs above the Uniqueness Threshold}

\author{
\begin{tabular}[h!]{cc}
   \FormatAuthor{Nicholas Kocurek}{nichok6@cs.washington.edu}{University of Washington}
   \FormatAuthor{Shayan Oveis Gharan}{shayan@cs.washington.edu}{University of Washington}
   \FormatAuthor{Dante Tjowasi}{dtjowasi@cs.washington.edu}{University of Washington}
\end{tabular}
} %

\date{\small\today}

\begin{document}
\maketitle
\allowdisplaybreaks
\begin{abstract}
    We design an efficient sampling algorithm to generate samples from the hardcore model on random regular bipartite graphs as long as $\lambda\lesssim \frac{1}{\sqrt{\Delta}}$, where $\Delta$ is the degree. Combined with recent work of Jenssen, Keevash and Perkins this implies an $\FPRAS$ for the partition function of the hardcore model on random regular bipartite graphs at any fugacity. Our algorithm is shown by analyzing two new Markov chains that work in complementary regimes. Our proof then proceeds by showing the corresponding simplicial complexes are top-link spectral expanders and appealing to the trickle-down theorem to prove fast mixing. 
\end{abstract}
\thispagestyle{empty}

\newpage

{
\hypersetup{hidelinks}
\tableofcontents
}
\thispagestyle{empty}
\newpage

\setcounter{page}{1}

\section{Introduction}
\label{sec:introduction}

A simplicial complex $\calX$ on a finite ground set $[n]=\{1,\dots,n\}$
is a downwards closed collection of subsets of $[n]$, i.e., if $\tau\in \calX$ and $\sigma \subseteq \tau$, then $\sigma \in \calX$. 
The elements of $\calX$ are called {\bf faces}, and the maximal faces are called {\bf facets}. 
We say that a face $\tau$ is of dimension $k$ if $\abs{\tau} = k$ and write $\text{dim}(\tau) = k$\footnote{Note that this differs from the typical topological definition of dimension for faces of a simplicial complex.}. 
A simplicial complex $\calX$ is a {\bf pure} $d$-dimensional complex if every facet has dimension $d$. 

Given a  $d$-dimensional complex $\calX$,
for any   $0 \leq i \leq d$,    define  $\calX(i) = \{\tau \in \calX \mid \text{dim}(\tau) = i\}$.  Moreover, the codimension of a face $\tau \in \calX$ is defined as $\codim (\tau)  = d - \dim (\tau)$. For a face $\tau \in \calX$, define the  {\bf link} of $\tau$ as the simplicial complex $\calX_{\tau} = \{\sigma \setminus \tau \mid \sigma \in \calX, \sigma \supseteq \tau\}$. 

Given a complex $(\calX,\mu)$, the (weighted) 1-skeleton of $\calX$, $G_\varnothing$, is the weighted graph with vertex set $\calX(1)$ and edge set $\calX(2)$ where the weight of an edge $$w(u,v)=\Pr_{\sigma\sim\mu}[u,v \in\sigma].$$
We let $\A_\varnothing$ be the adjacency matrix of this graph, and $\P_\varnothing$ be the transition probability matrix of the simple random walk. More generally, for any face $\tau\in \calX$ with $\codim(\tau)\geq 2$, we let $\A_\tau, \P_\tau$ be the adjacency matrix and the random walk matrix of the 1-skeleton of the link $\calX_\tau$. Note that for any $\tau \in \calX$, $\mu$ induces a probability distribution, $\mu_{|\tau}$ on the facets of $\calX_\tau$ where the probability of a facet $\sigma'$ is $\P_{\sigma\sim\mu}[\sigma'\subseteq \sigma \mid \tau \subseteq \sigma]$. Lastly, we say $\calX$ is a {\bf connected} complex if for every $\tau$ of codimension at least 2, the 1-skeleton of $\calX_\tau$ is a connected graph.

\begin{definition}[(Top-link) Spectral Expanders]
We say that $(\calX,\mu)$ is an $\alpha$-local spectral expander if for any link $\tau \in \calX$,
 $\lambda_2(\P_\tau)\leq \alpha$.
    We say that $(\calX,\mu)$ is a $\alpha$-top-link spectral expander if for any $\tau\in \calX$ with $\codim(\tau)=2$, $\lambda_2(\P_\tau)\leq \alpha$.
\end{definition}

Given a $d$-dimensional complex $(\calX,\mu)$ and a facet $\tau\in \calX(d)$, one can run a Markov chain called the  the down-up walk to generate random samples from $\mu$: each time step we choose $v \in \tau$  uniformly at random, and among all facets $\sigma\supseteq \tau \setminus \{v\}$ we choose one proportional to its weight $\mu(\sigma)$. It turns out that if $\calX$ is connected, this Markov chain converges to $\mu$. So, a natural question is to find sufficient conditions for the Markov chain to mix rapidly.

Over the last few years (top-link) high-dimensional expanders have been used extensively in the analysis of Markov chains.  The following local-to-global theorem is central in many such applications: 
\begin{theorem}[\cite{DinurK17, Oppenheim18,AlevL20}]\label{thm:oppenheims}
If $(\calX,\mu)$ is a connected $d$-dimensional $\frac{1-\delta}{d}$-top-link spectral expander then $\calX$ is a $\frac{1-\delta}{\delta d}$-local spectral expander. In particular, the down-up walk mixes in time $\poly(n^\delta,d,\min_\tau \log\frac{1}{\mu(\tau)}).$
\end{theorem}

\paragraph{Hardcore Model.} 
Given a graph $G= (V,E)$ and a parameter $\lambda>0$, let $\mu$  be the probability distribution over independent sets of $G$ where $\mu(I)= \frac{\lambda^{|I|}}{Z_G(\lambda)}$. We call $Z_G(\lambda) = \sum_{I \in \calI(G)} \lambda^{\abs{I}}$ the {\it partition function} and note exact computation of $Z_G(\lambda)$ is classically $\#\PTIME$-hard \cite{Valiant79B}. Recently, the above framework has been shown to be very successful in sampling from the hardcore model \cite{AnariLOG20, ChenLV20, ChenLV21, ChenFYZ21, BlancaCCPSV22, AnariJKPV22, ChenE22, ChenCYZ25, ChenCCYZ25}, where it is shown that the down-up walk mixes in polynomial time on graphs with maximum degree $\Delta$ so long as $\lambda \leq \lambda_c(\Delta) \coloneqq \frac{(\Delta-1)^{\Delta-1}}{(\Delta-2)^\Delta}$. This threshold, $\lambda_c(\Delta)$, has long been known as the {\it tree uniqueness threshold} \cite{Kelly85} and marks the critical threshold for which the Gibbs distribution for the hardcore model on the infinite $\Delta$-regular tree is unique if and only if $\lambda < \lambda_c(\Delta)$. It is further known that unless $\NP = \RP$, there exists no polynomial time algorithm to approximate $Z_G(\lambda)$ approximating $Z_G(\lambda)$ \cite{Sly10, SlyS12, GalanisGSVY14, GalanisSV15, GalanisSV16}, giving a near complete characterization of the computational complexity of computing $Z_G(\lambda)$, at least given polynomial time.

A fundamental open problem in the field is whether we can estimate $Z_G(\lambda)$ above the uniqueness threshold $\lambda_c(\Delta)$ when $G$ is a bipartite graph. When $\lambda = 1$, this task reduces to to count independent sets in a bipartite graph, aptly named $\BIS$. This problem, more generally for $\lambda > \lambda_c(\Delta)$, turns out to be a natural starting point for approximation-preserving reductions for a large class of intermediate problems known as $\RHPI$ \cite{DyerGGJ00}. Formally, this class $\RHPI$ has been shown to form an approximate counting trichotomy theorem (along with the classes $\FP$ and $\#\PTIME$) for approximately counting solutions to Boolean CSPs \cite{DyerGJ10}, akin to classical CSP dichotomy theorems. Many consider the conjecture of whether $\BIS$ is $\#\PTIME$-hard to be a counting and sampling analogue of the Unique Games Conjecture of \cite{Khot02}. Unlike the Unique Games Conjecture, we have little evidence to justify or refute this conjecture. Perhaps most notably, it is known that many variants of a natural chain called the {\it Glauber dynamics} mix only in exponential time even when $G$ is a uniformly random $\Delta$-regular graph and $\lambda$ is above the uniqueness threshold.

 \begin{theorem}[\cite{DyerFJ99, MosselWW07}]
    \label{thm:failtomix}
     Fix $\Delta \geq 3$ and let $G = (X,Y,E)$ be a random $\Delta$-regular bipartite graph. Then the Glauber dynamics (the down-up walk on the independent set complex) requires exponential-time in $\abs{X}+\abs{Y}$ to mix. Moreover, the same result holds for any $o(\abs{X}+\abs{Y})$-cautious Markov chain. A Markov chain is said to be $\ell$-cautious if it adds or deletes at most $\ell$ vertices in each step.
 \end{theorem}

Our main contribution is to design and analyze  Markov chains to estimate the partition function on random $\Delta$-regular bipartite graphs as long as $\lambda \lesssim \frac{1}{\sqrt{\Delta}}$. Together with the recent work of \cite{JenssenKP19} which works in the complementary regime $\lambda \gtrsim \frac{\log^2\Delta}{\Delta}$, this shows that one can estimate the partition function of random $\Delta$-regular bipartite graphs at any fugacity $\lambda$. We expect our findings to help better our understanding of the computational complexity of $\BIS$ in the future. 

 \subsection{Main results}

The following is our main theorem. 
 \begin{theorem}[Main]
    \label{thm:main}
    Let $G = (X, Y, E)$ be a random $\Delta$-regular bipartite graph. Then with high probability over the randomness of $G$, there is an $\FPRAS$ for the partition function $Z_G(\lambda)$ of the hardcore model on $G$ at fugacity $\lambda \lesssim \frac{1}{\sqrt{\Delta}}$.
\end{theorem}

We note that, to the best of our knowledge, there were no prior works analyzing the performance of MCMC algorithms for sampling from the hardcore model on (random regular) {\em bipartite} graphs above the uniqueness threshold; the only result which comes close is \cite{ChenGGPSV21}, an MCMC-based adaptation of \cite{JenssenKP19}, working in the high fugacity regime, i.e., when $\lambda \gg \Delta$ for expander graphs. The main result of \cite{JenssenKP19} also holds for random regular graphs, working when $\lambda \gtrsim \frac{\log^2\Delta}{\Delta}$, and by combining this with \Cref{thm:main} we immediately get an $\FPRAS$ for random $\Delta$-regular bipartite graphs at any fugacity $\lambda > 0$.

\begin{corollary}[\Cref{thm:main} + Theorem 2 of \cite{JenssenKP19}]
    Let $G = (X, Y, E)$ be a random $\Delta$-regular bipartite graph. Then with high probability over the randomness of $G$, there is an $\FPRAS$ for the partition function $Z_G(\lambda)$ of the hardcore model on $G$ at any fugacity $\lambda > 0$.
\end{corollary}
 
Our proof of \Cref{thm:main} goes through studying the down-up walk on a \emph{pair} of complexes closely related to the hardcore model, which we introduce here. The first is the {\it two-sided slice}, a simple modification where instead of considering the whole independent set complex, we truncate to sets with fixed sizes in both partitions of the graph.

\begin{definition}[Two-sided independent set slice]
    Let $G = (X, Y, E)$ be a bipartite graph and fix $k_X \leq \abs{X}, k_Y \leq \abs{Y}$. We define
    \begin{align*}
        \calI_{k_X, k_Y}(G) = \cbra{I \in \calI_k(G) \mid k = k_X + k_Y, \abs{I \cap X} = k_X, \abs{I \cap Y} = k_Y}\mper
    \end{align*}
    Define $\mu^{(k_X, k_Y)}$ to be the uniform distribution over $\calI_{k_X, k_Y}(G)$. We define the associated pure $(k_X + k_Y)$-dimensional simplicial complex $(\calX, \mu^{(k_X, k_Y)})$ by taking the closure of $\calI_{k_X, k_Y}(G)$ over ground set $X \cup Y$.
\end{definition}

The second complex we introduce is what we call the {\it one-sided slice} of the hardcore model. This distribution is over fixed-size subsets of just one side $\binom{X}{k}$ and weights each set proportional to the total weight of all independent sets with its intersection on $X$. 

\begin{definition}[One-sided slice of the hardcore model]
    Let $G = (X, Y, E)$ be a bipartite graph and fix $k \leq \abs{X}$. We define $\mu_\lambda^{(k)}$ over $\binom{X}{k}$ at fugacity $\lambda > 0$ via
    \begin{align*}
        \mu_\lambda^{(k)}(S) \propto \sum_{\substack{I \in \calI(G)\\ I \cap X = S}} \lambda^{\abs{I}}\mper
    \end{align*}
    for any $S \in \binom{X}{k}$. We define the associated pure $k$-dimensional simplicial complex $(\calX, \mu_\lambda^{(k)})$ by taking the complete complex on $\binom{X}{k}$ over $\mu_\lambda^{(k)}$.
\end{definition}

To the best of our knowledge, none of the previous works on the subject have studied these complexes and we expect our findings help in better understanding the complexity of $\BIS$. 

As our main technical result, we establish the near exact regimes for which these complexes are expanders for a random regular bipartite graph. Starting with the two-sided independent set slice, we prove the following.

\begin{theorem}
    \label{thm:bipartite}
    Let $G = (X, Y, E)$ be a random $\Delta$-regular bipartite graph. Then with high probability over the randomness of $G$, the two-sided independent set slice $(\calX, \mu^{(k_X, k_Y)})$ for $k_X, k_Y \leq \alpha \abs{X}$ with $\alpha =  \frac{\log \Delta}{(2+o_\Delta(1))\Delta}$ is a connected $\frac{1}{2(k_X + k_Y)}$-top-link spectral expander, and as a result, the down-up walk mixes in polynomial-time.
\end{theorem}

In other words, for sufficiently small $k_X, k_Y$ the two-sided slice expands enough to apply the trickle-down theorem of \Cref{thm:oppenheims}. The bound for $\alpha$ here turns out to be essentially tight for random regular graphs; beyond this point the complex is not only likely not an expander, it is not necessarily even connected. So instead of sampling using this chain, we show that above this bound the one-sided slice complex is complementary, in the sense that it is a top-link expander.

\begin{theorem}
    \label{thm:bipartite2}
    Let $G = (X, Y, E)$ be a random $\Delta$-regular bipartite graph. Then with high probability over the randomness of $G$, the one-sided slice of the hardcore model $(\calX, \mu_\lambda^{(k)})$ for $\alpha \abs{X} \leq k \leq \beta \abs{X}$ with $\alpha =  \frac{\log \Delta}{(2+o_\Delta(1))\Delta}$ and $\beta = \Theta(\lambda)$ at fugacity $\lambda \lesssim \frac{1}{\sqrt{\Delta}}$ is a connected $\frac{1}{2k}$-top-link spectral expander, and as a result, the down-up walk mixes in polynomial-time.
\end{theorem}
 
Importantly, the threshold $\alpha$ here matches that of \Cref{thm:bipartite} exactly. Finally, by applying the same techniques to the natural down-up walk for the independent set slice (see below) on a random $\Delta$-regular graph, we additionally get the following result.

\begin{theorem}
    \label{thm:regular}
    Let $G = (V, E)$ be a random $\Delta$-regular graph. Then with high probability over the randomness of $G$, the independent set slice $(\calX, \mu^{(k)})$ for $k \leq \alpha \abs{V}$ with $\alpha = \frac{\log \Delta}{(2+o_\Delta(1))\Delta}$ is a connected $\frac{1}{2k}$-top-link spectral expander, and as a result, the down-up walk mixes in polynomial-time.
\end{theorem}

\subsection{Related work and discussion}

\paragraph{Independent Set Slice.} To best understand our results, it is useful to introduce the notion of the independent set slice. Given a graph $G = (V,E)$, let $\calI_k(G)$ be the set of independent sets $I \in \calI(G)$ with $\abs{I} = k$ and let $\mu^{(k)}$ be the uniform distribution over $\calI_k(G)$. We can also define the $k$-dimensional complex $(\calX, \mu^{(k)})$ analogously. The history of the approximability of $|\calI_1(G)|,|\calI_2(G)|,\dots$ follows a similar course to the hardcore model, with the parameter of interest being the \emph{occupancy fraction} rather than the fugacity. The occupancy fraction of $G$ is the quantity $\alpha \coloneqq \E_{I \sim \mu}\sbra{\abs{I}}$ where $\mu$ is the hardcore model for $G$. Just as the infinite $\Delta$-regular tree plays an important role in defining the uniqueness threshold for fugacity, it also helps define a similar occupancy threshold as follows: while the Gibbs measure on the tree is not unique at fugacity $\lambda > \lambda_c(\Delta)$, there is a natural translation invariant Gibbs measure capturing most interesting behavior \cite{BhatnagarST14}. By defining $\alpha = \alpha(\lambda, \Delta)$ as the density of this measure we have the relation
\begin{align*}
    \lambda = \frac{\alpha}{1-2\alpha} \pbra{\frac{1-\alpha}{1-2\alpha}}^{\Delta-1}\mper
\end{align*}
We skip the exact derivation (see \cite{BhatnagarST14}) and instead highlight the values important to computational complexity. Similar to the hardcore model, $\alpha_c(\Delta)$, which is the value of $\alpha$ at the critical fugacity $\lambda_c(\Delta)$, marks a computational phase transition where for $k < \alpha_c(\Delta)$ there is an $\FPRAS$ for $|\calI_k(G)|$ and when $k > \alpha_c(\Delta)$ such a result would imply $\NP = \RP$ \cite{DaviesP21}. It was also shown in a sequence of recent works \cite{AlevL20, JainMPV23} that the down-up walk on the simplicial complex $(\calX, \mu^{(k)})$ mixes in near-linear time up to this threshold.

Up to this point, the approximability of the hardcore model and independent set slice have coincided cleanly. Forgoing the worst-case and instead assuming the graph $G$ of interest is a random regular graph, \cite{ChenCCYZ25} showed that the down-up walk on the standard independent set complex mixes in polynomial-time up to $\lambda \lesssim \frac{1}{\sqrt{\Delta}}$. As they point out, the corresponding occupancy fraction to this fugacity is $\alpha \approx \frac{\log \Delta}{2 \Delta}$. Our \Cref{thm:regular} then parallels their result, showing that the down-up walk mixes in polynomial-time up to the analogous threshold, significantly past the critical occupancy threshold. Moreover, our result goes through with a single application of standard trickle-down, whereas \cite{ChenCCYZ25} requires going through the field dynamics (a continuous analogue of the Glauber dynamics) and uses some heavy machinery in the form of a generalized trickle-down theorem for localization schemes \cite{AnariKV24}. However, we do concede that their result is \emph{derandomized}, in the sense that it works for any near-Ramanujan graph, whereas our result depends on certain combinatorial structure within random regular graphs.

Although our main result and proof techniques have some resemblance to the recent work of \cite{ChenCCYZ25}, we emphasize that the Glauber dynamics (analyzed in their work) does not mix on random regular bipartite graphs by \Cref{thm:failtomix}. A priori, it is unclear whether their techniques extend to our regime. Instead, to prove our main results, we redo their work analyzing the hardcore model on random regular graphs (\Cref{thm:regular}) using the independent set slice complex, and we manage to extend this approach to our defined slices of the bipartite case as well, proving our main theorem.

\subsection{Proof overview}

The proof of our main result, \Cref{thm:main}, follows primarily from two ideas: (i) we prove fast mixing of the down-up walks for the two-sided and one-sided independent set slices in complementary regimes for random regular bipartite graphs by using a combination of properties of random graphs and trickle-down theorems on high-dimensional expanders (\Cref{thm:oppenheims}) and (ii) we show a natural way to approximate the partition function of the hardcore model given approximate samplers for the aforementioned distributions.

    \paragraph{Part (i): Overview of \Cref{thm:bipartite}.} To prove fast mixing from \Cref{thm:oppenheims}, it is enough to establish $\lambda_2(\P_\tau) < \frac{1}{\abs{\tau}}$ for all links $\tau$ of codimension $2$. We start by showing how this process works for the two-sided independent set slice. Recall that after fixing $k_X \leq \abs{X}, k_Y \leq \abs{Y}$, the two-sided slice is the $(k_X+k_Y)$-dimensional complex $(\calX, \mu^{(k_X, k_Y)})$. 
    
    Since each facet has $k_X$ elements from $X$ and $k_Y$ from $Y$, there are two kinds of links of codimension 2: (1) those missing two elements from the same side and (2) those missing one element from each side. Showing the links of the former are expanders is straightforward as they correspond to a complete graph.

    So, we are left to bound the second eigenvalue for links of the latter type, which we end up doing by establishing
    \begin{equation}\label{eq:overviewgoal1}
        \lambda_2(\P_\tau) \lesssim \frac{\lambda_2(\A_G)}{\abs{(X \cup Y) \setminus (\tau \cup N\sbra{\tau})}}\mper
    \end{equation}
    
    To see why \eqref{eq:overviewgoal1} is sufficient, we show that both of these quantities are well-controlled in random regular graphs. First, since random regular graphs are near-Ramanujan we have $\lambda_2(\A_G) \lesssim 2 \sqrt{\Delta-1}$. The question then is what $\abs{(X \cup Y) \setminus (\tau \cup N\sbra{\tau})}$ looks like as a function of $\abs{\tau}$. Our primary observation is that when $\abs{\tau_X} \leq \frac{\log \Delta}{(2+o_\Delta(1))\Delta}\abs{X}$, with high probability we have $\abs{Y \setminus N\sbra{\tau_X}} \gtrsim \frac{\log \Delta}{\sqrt\Delta} \abs{Y}$ for all such $\tau$ (see \Cref{lem:bipartiteconc}). The proof of this fact uses a martingale concentration argument for the pairing model. Plugging these two values in above then gives the bound of $\frac{1}{\abs{\tau}}$ as desired.
    
    The key observation towards \eqref{eq:overviewgoal1} is then that for links of the latter type, the adjacency matrix satisfies
    \begin{align*}
        \A_\tau = \overline{\A_{G\sbra{(X \cup Y) \setminus (\tau \cup N\sbra{\tau})}}}\mcom
    \end{align*}
    where the complement here is the bipartite complement, given by flipping only the edges crossing $X$ and $Y$. This matrix is simply what is left over after removing $\tau$ and its neighbors from $G$, with the bipartite complement being taken because independent sets in these links require non-edges across partitions. To get from here to the bound in \eqref{eq:overviewgoal1}, we simply observe the eigenvalues of $\A_\tau$ interlace $\overline{\A_G}$ and argue $\lambda_2(\overline{\A_G}) \leq \lambda_2(\A_G)$.

    While this result works for $\abs{\tau} \leq \frac{\log \Delta}{(2+o_\Delta(1)) \Delta}$, which we recall is the occupancy fraction corresponding to $\lambda = O\pbra{\frac{1}{\sqrt{\Delta}}}$, it is not clear how to push it beyond this threshold. Indeed, we believe the size of $\abs{(X \cup Y) \setminus (\tau \cup N\sbra{\tau})}$ begins to shrink quickly beyond this threshold, leading to links possibly becoming disconnected.

    \paragraph{Part (i): Overview of \Cref{thm:bipartite2}.} To get around this, we take a completely different route above the critical occupancy threshold; namely, we show that a very different simplicial complex, the one-sided slice $(\calX,\mu_\lambda^{(k)})$ is a top-link expander as long as $k\geq \frac{\log\Delta}{(2+o_\Delta(1))\Delta}|X|$ and $\lambda\lesssim \frac{1}{\sqrt{\Delta}}$. Links of codimension 2 of this complex are in one sense easier to analyze than in the previous complex, as their support is always a complete graph with exactly $|X|-k$ many vertices. The difficulty, however, is that this is now a weighted graph. It turns out that for a link $\tau$ (of codimension 2), the adjacency matrix satisfies
    \begin{align*}
        \A_\tau(u, v) = (1+\lambda)^{-\abs{N_\tau(u)}-\abs{N_\tau(v)}+\abs{N_\tau(u) \cap N_\tau(v)}}\mcom
    \end{align*}
    where $N_\tau(u)$ is the set of neighbors of $u$ in the graph $G \setminus N[\tau]$, where $N[\tau]$ is the set of neighbors of $\tau$.
    For a random regular graph, $|N_\tau(u)\cap N_\tau(v)|\in \{0,1\}$ with high probability; however, $|N_\tau(\cdot)|$ could still vary in the range $[0,\Delta]$. In our regime of $\lambda\lesssim\frac{1}{\sqrt{\Delta}}$ we mainly care whether $|N_\tau(u)|\ll \sqrt{\Delta}$, for which $(1+\lambda)^{\abs{N_\tau(u)}}$ is constant, and our main observation is that with high probability for all $\tau$, almost all vertices $u\in \calX_\tau(1)$ indeed satisfy $|N_\tau(u)| \lesssim \sqrt{\Delta}$. This simplifies the structure of $\A_\tau$ as we can show that (up to a normalization) $\lambda_{\max}(\A_\tau - \mathbf{1}\mathbf{1}^\top) \leq \lambda \cdot \lambda_2(\A_G^2)$ using an interlacing argument, which leads to a bound $\lambda_2(\P_\tau)<\frac{1}{|\tau|}$.

\paragraph{Part (ii): Overview of \Cref{thm:main}.} Assuming \Cref{thm:bipartite} and \Cref{thm:bipartite2}, we prove \Cref{thm:main}.
\begin{proof}[Proof of \Cref{thm:main}]
 We start by defining an approximation to the partition function. Fix $\alpha = \frac{\log \Delta}{(2+o_\Delta(1)) \Delta}$ and $\beta = \Theta(\lambda)$ as in \Cref{thm:bipartite} and \Cref{thm:bipartite2} and define
    \begin{align*}
        \hat{Z}_G(\lambda) = \sum_{\substack{I \in \calI(G)\\ \abs{I \cap X} \leq \alpha \abs{X}\\ \abs{I \cap Y} \leq \alpha \abs{Y}}} \lambda^{\abs{I}} + \sum_{\substack{I \in \calI(G)\\ \alpha \abs{X} < \abs{I \cap X} \leq  \beta\abs{X}}} \lambda^{\abs{I}} + \sum_{\substack{I \in \calI(G)\\ \alpha \abs{Y} < \abs{I \cap Y} \leq  \beta\abs{Y}}} \lambda^{\abs{I}}\mper
    \end{align*}

    Our goal is to show that with high probability over the randomness $G$ we have
    \begin{equation}\label{eq:approx}
        \abs{Z_G(\lambda)-\hat{Z}_G(\lambda)} \leq \exp(-(\abs{X} + \abs{Y})) \cdot Z_G(\lambda)\mper
    \end{equation}
    Given this, an $\FPRAS$ for $\hat{Z}_G(\lambda)$ immediately implies an $\FPRAS$ for $Z_G(\lambda)$. This is because if $\exp(-(\abs{X} + \abs{Y})) < \frac{\varepsilon}{2}$, we simply approximate $\hat{Z}_G(\lambda)$ to accuracy $1+\frac{\varepsilon}{2}$, yielding a $(1+\varepsilon)$-approximation to $Z_G(\lambda)$. If $\frac{\varepsilon}{2} < \exp(-(\abs{X}+\abs{Y}))$ we can brute force $Z_G(\lambda)$ directly.

    We now prove \eqref{eq:approx} holds. Observe that by partitioning the partition function based on the thresholds $\alpha, \beta$ and applying triangle inequality we have
    \begin{align*}
        \abs{Z_G(\lambda)-\hat{Z}_G(\lambda)} \leq \sum_{\substack{I \in \calI(G)\\ \abs{I \cap X} > \beta \abs{X}\\ \abs{I \cap Y} > \beta \abs{Y}}} \lambda^{\abs{I}} + \sum_{\substack{I \in \calI(G)\\ \alpha \abs{X} < \abs{I \cap X} \leq  \beta\abs{X}\\ \alpha \abs{Y} < \abs{I \cap Y} \leq  \beta\abs{Y} }} \lambda^{\abs{I}}\mcom
    \end{align*}

    since $\hat{Z}_G(\lambda)$ drops any independent set exceeding occupancy $\beta$ within $X$ or $Y$ and double counts any independent set with occupancy in $[\alpha, \beta]$ in both. The bound in \eqref{eq:approx} then follows immediately from the following lemmas, which show that the weight of the two terms above is exponential small in $Z_G(\lambda)$.

    \begin{lemma}
        \label{lem:indsetbound}
        Let $G = (V, E)$ be a graph and $\mu$ the corresponding hardcore model at fugacity $\lambda > 0$. Then,
        \begin{align*}
            \Pr_{I \sim \mu}\sbra{\abs{I} \geq 4\lambda \abs{V}} \leq \exp(-\abs{V})\mper
        \end{align*}
    \end{lemma}

    \begin{lemma}
        \label{prop:bipartiteindsetbound}
        Let $G = (X, Y, E)$ be a random $\Delta$-regular bipartite graph and $\mu$ the corresponding hardcore model at fugacity $\lambda \lesssim \frac{1}{\sqrt{\Delta}}$ and let $\alpha = \frac{\log \Delta}{(2+o_\Delta(1))\Delta}$. Then with high probability over the randomness in $G$,
        \begin{align*}
            \Pr_{I \sim \mu}\sbra{\abs{I \cap X} > \alpha\abs{X}, \abs{I \cap Y} > \alpha \abs{Y} } \leq \exp(-(\abs{X} + \abs{Y}))\mper
        \end{align*}
    \end{lemma}

    We delay the proofs of the above to \Cref{sec:bipartiteprops} and finish by arguing that there is an $\FPRAS$ for $\hat{Z}_G(\lambda)$. We suggestively wrote $\hat{Z}_G(\lambda)$ in three parts and will argue each has an $\FPRAS$ separately, from which the result follows. Note that we can rewrite the first term as
    \begin{align*}
        \sum_{\substack{I \in \calI(G)\\ \abs{I \cap X} \leq \alpha \abs{X}\\ \abs{I \cap Y} \leq \alpha \abs{Y}}} \lambda^{\abs{I}} = \sum_{k_X = 1}^{\alpha\abs{X}} \sum_{k_Y = 1}^{\alpha \abs{Y}} \abs{\calI_{k_X,k_Y}(G)} \cdot \lambda^{k_X + k_Y} \mper
    \end{align*}
    Since there are only polynomial many terms in the latter sum, it suffices to have an $\FPRAS$ for each $\abs{\calI_{k_X, k_Y}(G)}$, the two-sided slice of $G$, within the parameter regime of \Cref{thm:bipartite}. Similarly, we observe the second and third terms of $\hat{Z}_G(\lambda)$ are just the partition function for the one-sided slice of $G$ from $X$ and $Y$ respectively, within the parameter regime of \Cref{thm:bipartite2}, so an $\FPRAS$ for each suffices. To finish then, we just need to convert the approximate samplers given by \Cref{thm:bipartite} and \Cref{thm:bipartite2} to approximate counters.

    Reducing counting to sampling here turns out to be a bit subtle. The standard reduction \cite{JerrumVV86} does not apply black-box in our setting, since the problem is not self-reducible under the random $\Delta$-regular assumption. Luckily, the simplicial complex view of \Cref{thm:bipartite} and \Cref{thm:bipartite2} provide a natural fix, as the conditional distributions arising in the standard reduction correspond to the distributions over links in the simplicial complex. By the local-to-global arguments in \Cref{thm:oppenheims}, we also get polynomial-time mixing of the down-up walk within each link. We can then estimate the probability of any heavy element using standard Monte Carlo methods, recursing on the conditional each time, which suffices for the standard reduction from $\FPAUS$ to $\FPRAS$.
\end{proof}

\subsection{Structure of the paper}

The rest of the paper is organized as follows. In \Cref{sec:bipartiteprops}, we prove a collection of properties of random regular bipartite graphs to be used throughout the paper. In \Cref{sec:bipartite1}, we prove our first main result, \Cref{thm:bipartite}, that the two-sided independent set slice is a top-link spectral expander up to occupancy $\alpha = \frac{\log \Delta}{(2+o_\Delta(1))\Delta}$. In \Cref{sec:bipartite2}, we prove our second main result, \Cref{thm:bipartite2}, that the one-sided slice of the hardcore model is a top-link spectral expander beyond occupancy $\alpha$. Finally, in \Cref{sec:regular}, we show \Cref{thm:regular}, that the standard independent set slice is a top-link spectral expander up to $\alpha$ and in \Cref{sec:slowmix} we show a slow mixing example for the one-sided slice.

\subsection{Acknowledgements}
Our research is supported by NSF grant CCF-2203541, a Simons Investigator Award 928589, and a Lazowska Endowed Professorship in Computer Science \& Engineering. This work was also supported by NSF CAREER award IIS2541127.

\section{Preliminaries}
\label{sec:preliminaries}

\subsection{Basic notation}

For a rectangular matrix $A \in \R^{m \times n}$, let $A^\top$ denote its transpose. We let $\norm{A}{2} \coloneqq \max_{\substack{x \in \R^m, y \in \R^n\\ \norm{x}{2} = \norm{y}{2} = 1}} x^\top A y$ denote the spectral norm of $A$. For any $S \subseteq [n]$, we write $\Id_S \in \R^{n\times n}$ to be diagonal matrix $\Id_S(i,i) = 1$ for $i \in S$ and $\Id_S(i,i) = 0$ otherwise. Given a graph $G = (V, E)$ and a subset $S \subseteq V$, we define the set of neighbors of $S$ by the set $N\sbra{S} = \{T \subseteq V \setminus S \mid \forall u \in T, \exists v \in S, \{u,v\} \in E\}$. For sets $S, \tau \subseteq V$, we let $N_\tau\sbra{S}$ denote $N\sbra{S} \setminus N\sbra{\tau}$. For shorthand we write $N(v) \coloneqq N\sbra{\{v\}}$ for $v \in V$.

\subsection{Graph theory}

Throughout we assume all graphs are simple and refer to any graph with multiedges as a multigraph. We consider the following process for sampling a $\Delta$-regular multigraph, commonly called the pairing model. We start by defining some notation.

\begin{definition}[$\frac{1}{\Delta}$-vertices]
Given a graph $G = (V,E)$, a $\frac{1}{\Delta}$-vertex is an element of $V \times [\Delta]$ and we call $(v, i)$ the $i$th copy of $v$. We also refer to the set $C(v) \coloneqq \cbra{(v,i) \mid i \in [\Delta]}$, the set of all copies, as the cloud of $v$ and use $C(S) \coloneqq \bigsqcup_{v \in S} C(v)$ for $S \subseteq V$.
\end{definition}

The notion of $\frac{1}{\Delta}$-vertices allows us to sample a random regular graph by simply sampling a perfect matching on $\Delta$ copies of each vertex.

\begin{enumerate}
    \item Consider the set $V \times [\Delta]$ and assume it has even cardinality.
    \item Sample a perfect matching $\pi$ on the set of $\frac{1}{\Delta}$-vertices, $V \times [\Delta]$.
    \item Construct the induced multigraph $G$ on $V$ by doing the following:  for every $\{(u,i), (v,j)\} \in \binom{V \times [\Delta]}{2}$ matched in $\pi$, add the edge $\{u,v\}$ to $G$.
\end{enumerate}

By conditioning on the pairing model not giving any self-loops or multiedges in $G$, this process gives the uniform distribution on simple $\Delta$-regular graphs on $V$, which is what we mean when we say $G$ is a random $\Delta$-regular graph.

A similar model exists for sampling a uniformly random $\Delta$-regular bipartite graph.

\begin{definition}[Pairing model for random regular bipartite graphs] 

The pairing model for random regular bipartite graphs is a graph generated by the following process:

\begin{enumerate}
    \item For $X = [n/2]$ and $Y = [n/2]$, consider the sets $X \times [\Delta]$ and $Y \times [\Delta]$.
    \item Sample a perfect matching on $\frac{1}{\Delta}$-vertices $X \times [\Delta]$ and $Y \times [\Delta]$.
    \item Construct the induced bipartite multigraph on $G = (X, Y)$ by adding $\{x, y\}$ for every matched pair $\{(x,i), (y,j)\} \in X \times [\Delta] \times Y \times \Delta$.
\end{enumerate}

Once again conditioning on the pairing model not giving multiedges, this process gives the uniform distribution on $\Delta$-regular bipartite graphs on $[n]$. 
\end{definition}

By arguing the probability the pairing model gives a simple graph can be bound independent of $n$, \cite{Friedman03} established the following result on the near-Ramanujan property of random regular graphs.

\begin{theorem}[\cite{Friedman03}]
    \label{thm:friedman}
    Let $G = (V, E)$ be a random $\Delta$-regular graph on $n$ vertices. Then with high probability $|\lambdamin(\A_G)| \leq 2\sqrt{\Delta-1} + o(1)$.
\end{theorem}

An analogous result holds for random regular bipartite graphs.
\begin{theorem}[\cite{Friedman03}]
    \label{thm:friedman2}
    Let $G = (X, Y, E)$ be a random $\Delta$-regular bipartite graph on $n$ vertices. Then with high probability $\lambda_2(\A_G) \leq 2\sqrt{\Delta-1}+o(1)$.
\end{theorem}

In general, any with high probability (in $n$) event in the pairing model can be transferred to the random regular graph or bipartite graph case respectively using contiguity \cite{Wormald99}.

\begin{definition}[Contiguity]
    Let $\{\mu_n\}_{n \in \N}, \{\nu_n\}_{n \in \N}$ be two families of probability distributions parameterized by $n$. We say $\{\mu_n\}_{n \in \N}$ and $\{\nu_n\}_{n \in \N}$ are contiguous if any event $E$ happening with high probability in one also occurs with high probability in the other.
\end{definition}

\begin{fact}
    \label{fact:contiguity}
    The pairing model for $n, \Delta$ is contiguous to the random $\Delta$-regular graph distribution on $n$ vertices, and the bipartite pairing model for $n, \Delta$ is contiguous to the $\Delta$-regular bipartite graph distribution on $n$ vertices.
\end{fact}

The upshot of this is that we may prove with high probability statements in the simpler pairing models in order to recover them for random regular graphs. See \cite{Puder15} for more details and references on contiguity.

\subsection{Linear algebra}

We recall here forms of the Cauchy Interlacing Theorem that will be useful throughout.

\begin{fact}[Cauchy Interlacing Theorem]
    \label{fact:cauchyinterlacing}
    Let $A \in \R^{n \times n}$ be a symmetric matrix. Let $B = \Pi A \Pi$ for some rank-$k$ projector matrix $\Pi \in \R^{n \times n}$. Then the decreasing-ordered eigenvalues satisfy
    \begin{align*}
        \lambda_i(A) \geq \lambda_i(B) \geq \lambda_{i+n-k}(A) \;\;\;\; \text{for } i = 1, \dots, k\mper
    \end{align*}
    In particular, if $A \succeq 0$ then $A \succeq B$.
\end{fact}

\begin{corollary}
    \label{cor:cauchyinterlacing}
    Let $A, B \in \R^{n \times n}$ be symmetric matrices and let $v \in \R^n$ be a vector. Then $A - vv^\top \preceq B$ implies $\lambda_2(A) \leq \lambda_1(B)$.
\end{corollary}

\subsection{Markov chains}

We provide here some background on Markov chains, for more details and references see \cite{LevinPW17}. A (discrete-time) Markov chain is a stochastic process $\{X_t\}_{t \in \Z_{\geq 0}}$ on a state set $\Omega$ satisfying the Markovian property $\Pr\sbra{X_{t+1} = x \mid X_t, ..., X_0} = \Pr\sbra{X_{t+1} = x \mid X_t}$ for every $t \in \Z_{\geq 0}$, $x \in \Omega$. We describe a Markov chain by $\P \in \R^{\Omega \times \Omega}$ where for $x, y \in \Omega$, $\P(x,y) = \Pr\sbra{X_{t+1} = y \mid X_t = x}$.

We say $\P$ is \textit{irreducible} if the (weighted) graph on $\Omega$ induced by $\P$ is connected. We say $\P$ is \textit{reversible} with respect to a distribution $\mu$ over $\Omega$ if $\mu(x)\P(x,y)  = \mu(y)\P(y,x)$ for all $x, y \in \Omega$. In this case we say $\mu$ is the stationary distribution for $\P$, satisfies $\mu\P = \mu$, and is always unique for irreducible, reversible $\P$.

Our key quantity of interest for Markov chains is the mixing time.

\begin{definition}[Mixing time]
    Let $\P$ be a Markov chain on $\Omega$ with stationary distribution $\mu$. For any distribution $\pi$ on $\Omega$, we define the $\varepsilon$-mixing time of $x \in \Omega$ with respect to $\P$ as
    \begin{equation*}
        \tau_x(\P, \varepsilon) = \min_{t \in \Z_{\geq 0}} \cbra{\dtv(\delta_x\P^t , \mu) \leq \varepsilon} \mcom
    \end{equation*}
    where $\delta_x$ is the point mass on $x$. The mixing time of $\taumix(\P, \varepsilon)$ is then $\max_{x \in \Omega} \tau_x(\P, \varepsilon)$.
\end{definition}

It is well-known that the mixing time is characterized by the existence of a spectral gap, that is, a lower bound on the quantity $1-\lambda_*$ where $\lambda_*$ is the second largest eigenvalue of $\P$ in magnitude.

\begin{fact}[Mixing via spectral gap]
    Let $\P$ be a Markov chain on $\Omega$ with stationary distribution $\mu$. Denote its eigenvalues $1 \geq \lambda_1 \geq \lambda_2 \geq \dots \lambda_n \geq -1$. Let $\lambda_* = \max\{\lambda_2, \abs{\lambda_n}\}$. Then the mixing time of $\P$ satisfies
    \begin{align*}
        \taumix(\P, \varepsilon) \leq O\left(\frac{1}{1-\lambda_*} \log\frac{1}{\varepsilon \cdot \min_{x \in \Omega} \mu(x)}\right)\mper
    \end{align*}
\end{fact}

We can standardly eliminate the dependence on $\abs{\lambda_n}$ in the bound above via the following observation.

\begin{fact}
    \label{fact:mixing}
    Let $\P$ be a Markov chain on $\Omega$ with stationary distribution $\mu$ with second largest eigenvalue $\lambda_2$. Then the lazy version of $\P$, defined as $\frac{\id + \P}{2}$, has stationary distribution $\mu$ and mixing time satisfying
    \begin{align*}
        \taumix\left(\frac{\id + \P}{2}, \varepsilon\right) \leq O\left(\frac{1}{1-\lambda_2} \log\frac{1}{\varepsilon \cdot \min_{x \in \Omega} \mu(x)}\right)\mper
    \end{align*}
\end{fact}

\subsection{Approximate counting and sampling}
We use the following two definitions to describe approximate counting and sampling algorithms.

\begin{definition}[$\FPRAS$]
    Given a finite set $\Omega$ and a weight function $w : \Omega \to \R^+$ of length $n$, define the partition function $Z \coloneqq \sum_{x \in \Omega} w(x)$. We say an algorithm $\calA$ is a fully polynomial randomized approximation scheme ($\FPRAS$) for $Z$ if given error parameter $\varepsilon \in (0,1)$ and confidence interval $\delta \in (0,1)$:
    \begin{itemize}
        \item Outputs $\hat{Z} \in \R^+$ satisfying $\Pr\sbra{(1-\varepsilon)Z \leq \hat{Z} \leq (1+\varepsilon) Z} \geq 1-\delta$.
        \item Runs in $\poly(n, \frac{1}{\varepsilon}, \log \frac{1}{\delta})$-time.
    \end{itemize}
\end{definition}

\begin{definition}[$\FPAUS$]
    Given a finite set $\Omega$ and a weight function $w : \Omega \to \R^+$ of length $n$, define the distribution $\mu$ on $\Omega$ via $\mu(x) = \frac{w(x)}{Z}$ where $Z$ is the partition function. We say an algorithm $\calA$ is a fully polynomial almost-uniform sampler ($\FPAUS$) for $\mu$ if given error parameter $\varepsilon \in (0,1)$:
    \begin{itemize}
        \item Outputs $x \sim \hat{\mu}$ satisfying $\dtv(\hat{\mu}, \mu) \leq \varepsilon$.
        \item Runs in $\poly(n, \log\frac{1}{\varepsilon})$-time.
    \end{itemize}
\end{definition}

For a large class of problems known as self-reducible problems there is an equivalence between these two notions. See \cite{JerrumVV86} for details.

\subsection{High-dimensional expanders}

In this section, we introduce some definitions for walks in high-dimensional expanders useful in applying \Cref{thm:oppenheims}. For most definitions and notation on simplicial complexes, see \Cref{sec:introduction}.

\begin{definition}[Local walk operators]
    Given a $d$-dimensional simplicial complex $\calX$ with distribution $\mu$ on $\calX(d)$, we define the local walk operator for $\tau \in \calX$ with $\codim(\tau) \geq 2$ as $\P_\tau \in \R^{\calX_\tau(1) \times \calX_\tau(1)}$ with entry $u, v \in \calX_\tau(1)$ being
    \begin{align*}
        \P^{(\mu)}_\tau(u,v) = \frac{1}{\codim(\tau)-1} \Pr_{S \sim \mu_{\mid \tau \cup \{u\}}}\sbra{v \in S}\mper
    \end{align*}
\end{definition}

One should think of the local walk as the natural graph random walk on the skeleton of the corresponding link induced by the global distribution $\mu$. Given such a $\mu$, we also define the following induced distribution $\pi^{(\mu)}_\tau$ over $\calX_\tau(1)$ for $\tau \in \calX$ with $\codim(\tau) \geq 2$:
\begin{align*}
    \pi^{(\mu)}_\tau(u) &= \frac{1}{\codim(\tau)} \Pr_{S \sim \mu_{\mid \tau}}\sbra{u \in S}\mper
\end{align*}
We point out that with this definition $\P_\tau(u,v) = \pi_{\tau \cup \{u\}}(v)$. Finally, we define the matrix $\Pi^{(\mu)}_\tau \in \R^{\calX_\tau(1) \times \calX_\tau(1)}$ by $\Pi^{(\mu)}_\tau = \diag(\pi^{(\mu)}_\tau)$. When clear from context, we drop the superscript $(\mu)$ in these operators.
\section{Properties of Random Regular Graphs}
\label{sec:bipartiteprops}

In this section, we prove a collection of properties of random regular (bipartite) graphs and the hardcore model which we reference throughout.

\subsection{Neighborhood concentration in random regular bipartite graphs}

\begin{lemma}
    \label{lem:bipartiteconc}
    Let $G=(X,Y,E)$ be a random $\Delta$-regular bipartite graph. For any $\ell > 0$, there exists $\gamma = o_\Delta(1)$ such that with high probability over the randomness of $G$, every $\tau \subseteq X$ satisfies the following:
    \begin{enumerate}
        \item \textnormal{(Expansion)} If $\abs{\tau} \geq \frac{\log \Delta}{(2+\gamma) \Delta} \abs{X}$ then
        \begin{align*}
            \abs{Y \setminus N\sbra{\tau}} \leq \frac{(\ell+1)\log \Delta}{\sqrt{\Delta}} \abs{Y}\mper
        \end{align*}
        \item \textnormal{(Anti-expansion)} If $\abs{\tau} \leq \frac{\log \Delta}{(2+\gamma)\Delta} \abs{X}$ then
        \begin{align*}
            \abs{Y \setminus N\sbra{\tau}} \geq \frac{\ell\log \Delta}{\sqrt{\Delta}}\abs{Y}\mper
        \end{align*}
    \end{enumerate}
\end{lemma}

\begin{proof}
   It suffices to prove this statement for all sets $\tau$ of size $\frac{\log \Delta}{(2+\gamma) \Delta} \abs{X}$, since the addition or removal of vertices causes the neighborhood to grow or shrink respectively. 
    
    First, by \Cref{lem:singlebipartite}
    \begin{equation}\E\sbra{\abs{Y \setminus N\sbra{\tau}}} = \frac{1 - o_\Delta(1)}{\Delta^{1/(2+\gamma)}}\abs{Y} =\frac{2\ell+1}{2} \cdot \frac{\log \Delta}{\sqrt{\Delta}}\abs{Y} \mcom
    \label{eq:EYNtau}
    \end{equation}
    where $\gamma=\Theta(\frac{\log\log \Delta + \log\ell}{\log \Delta})$ is chosen such that the equality holds.
    
    Our goal is now to prove concentration of $|N\sbra{\tau}|$ around this expectation (which is equivalent to proving concentration of $|Y \setminus N\sbra{\tau}|$), sufficient enough to union bound across all $\tau \subseteq X$ with $\abs{\tau} = \frac{\log \Delta}{(2+\gamma) \Delta} \abs{X}$.
    
    We start by fixing $\abs{\tau} = \frac{\log \Delta}{(2+\gamma) \Delta}\abs{X}$ and bounding the total number of sets in $\binom{X}{\abs{\tau}}$ by
    \begin{equation}\label{eq:binomest}
        \binom{\abs{X}}{\abs{\tau}} \leq \pbra{\frac{e\abs{X}}{\abs{\tau}}}^{\abs{\tau}} \leq \pbra{\frac{e(2+\gamma)\Delta}{\log \Delta}}^{\abs{\tau}} \leq \exp\pbra{\frac{\log^2 \Delta}{(2+\gamma)\Delta} \abs{X}}\mcom
    \end{equation}
    using standard binomial estimates and the assumption $\gamma < 1$. Using \eqref{eq:EYNtau} and \eqref{eq:binomest} with the following lemma will be sufficient to finish.

    \begin{lemma}
        \label{lem:martingaletail}
        For any $t>0$, we have
        $$\Pr\sbra{ \abs{|N[\tau]| - \E\sbra{|N[\tau]|}} > t}  \leq 2\exp\pbra{ - \frac{1}{2} \cdot \frac{t^2 }{\pbra{\frac{1}{8} + o_\Delta(1)}|Y| + t/3} }\mper$$
    \end{lemma}

    To see why, we plug $t =  \frac{\log \Delta}{2\sqrt{\Delta}}\abs{Y}$ into \Cref{lem:martingaletail} and apply a union bound over the entirety of $\binom{X}{\abs{\tau}}$ using our estimate \eqref{eq:binomest} to yield
    \begin{align*}
         \Pr\sbra{\forall \tau , \, \abs{\abs{N\sbra{\tau}} - \E\sbra{|N\sbra{\tau}|}} \geq  \frac{\log \Delta}{2\sqrt{\Delta}} |Y| } &\underset{\eqref{eq:binomest}}{\leq}  2\exp\pbra{\frac{\log^2 \Delta}{(2+\gamma) \Delta} \abs{X}} \cdot \exp\pbra{ - (1 - o_\Delta(1)) \cdot \frac{\log^2 \Delta}{\Delta  }|Y| } \\
         & \underset{|X|=|Y|}{\leq} \exp\pbra{ - \pbra{\frac{1}{2} - o_\Delta(1)} \cdot \frac{\log^2 \Delta}{\Delta  }|Y| }\mper
    \end{align*}

Conditioned on this event not happening, we have that for all such $\tau$, $|Y \setminus N[\tau]|$ satisfies
\begin{align*}
    \abs{Y \setminus N[\tau]} \underset{\eqref{eq:EYNtau}}{\geq} \left(\frac{2\ell+1}{2} - \frac{1}{2} \right) \frac{\log \Delta}{\sqrt{\Delta}}\abs{Y} \geq \frac{\ell\log \Delta}{\sqrt{\Delta}} \abs{Y}\mcom
\end{align*}
and similarly 
\begin{align*}
    \abs{Y \setminus N[\tau]} \underset{\eqref{eq:EYNtau}}{\leq}  \left(\frac{2\ell+1}{2} + \frac{1}{2} \right) \frac{\log \Delta}{\sqrt{\Delta}} \abs{Y} \leq \frac{(\ell+1)\log \Delta}{\sqrt{\Delta}}\abs{Y}\mcom
\end{align*}
as desired.
\end{proof}

\begin{proof}[Proof of \Cref{lem:martingaletail}]
    Define the following edge exposure martingale for the pairing model. Perhaps after renaming, we assume the $\frac{1}{\Delta}$-vertices in the cloud $C(\tau)$ are numbered $1,\dots,\Delta|\tau|$. For any such $i = 1, \dots, \Delta\abs{\tau}$ then, let $\pi_i \in Y \times [\Delta]$  be the match for $i$. Then, define
    \begin{align}\label{eq:Zidef}
        Z_i \coloneqq \E_{\pi_{i+1}, \dots, \pi_{\Delta\abs{\tau}}}\sbra{|N_\pi\sbra{\tau}| \mid \pi_{i}, \dots, \pi_1}\mper
    \end{align}
    Observe that $Z_0 = \E \sbra{|N\sbra{\tau}|}$ while $Z_{\Delta\abs{\tau}} = \abs{N\sbra{\tau}}$, i.e.~it is the number of neighbors of $\tau$ in this realization of the pairing model. As such, it is enough to show that
    \begin{equation}
        \Pr\sbra{|Z_{\Delta|\tau|} - Z_0| \geq t}\leq 2\exp\pbra{ - \frac{1}{2} \cdot \frac{t^2 }{\pbra{\frac{1}{8} + o_\Delta(1)}|Y| + t/3} }\mper
    \end{equation}
    To do this, we use Freedman's concentration inequality for martingales.
    \begin{lemma}[Freedman's inequality \cite{Freedman75}]
        Let $\{\Delta_i \coloneqq Z_{i} - Z_{i-1}\}_{i \in [m]}$ be a martingale difference sequence satisfying $\abs{\Delta_i} \leq 1$ with probability 1 for every $i \in [m]$. Then for all $t, v>0$ we have
        \begin{align*}
            \Pr\sbra{|Z_m- Z_0| > t \text{ and } V \leq  v } \leq 2\exp\pbra{- \frac{ t^2}{2(v + t/3)}}\mcom
        \end{align*}
        where $V = \sum_{i = 1}^m \E\sbra{(Z_i - Z_{i-1})^2 \mid Z_{i-1}}$ is the predictable quadratic variation of $\{Z_i\}$.
    \end{lemma}
To apply the inequality, we first establish the following lemma for our specific martingale.
\begin{lemma}
    \label{lem:martingalediff}
    For the Doob martingale $\{Z_i\}_{1\leq i\leq \Delta|\tau|}$ defined in \eqref{eq:Zidef} and all $1\leq i\leq \Delta|\tau|$ we have
    \begin{enumerate}
        \item $|Z_i-Z_{i-1}|\leq 1$ with probability 1.
        \item $\E\sbra{(Z_i-Z_{i-1})^2|Z_{i-1}} \leq \frac{1}{4} \pbra{1-\frac{\Delta-1}{\Delta|X|}}^{2(\Delta|\tau|-i)}$ with probability 1.
    \end{enumerate}
\end{lemma}
Now, by the second conclusion of \Cref{lem:martingalediff}, we have with probability 1,
\begin{align*}
    \sum_{i=1}^{\Delta|\tau|} \E\sbra{(Z_i-Z_{i-1})^2 \mid Z_{i-1}} &\leq  \frac{1}{4} \sum_{i = 1}^{\Delta|\tau| } \pbra{1 - \frac{\Delta - 1}{\Delta |Y|}}^{2 (\Delta|\tau| - i)} \\
        &\leq \frac14\cdot \frac{1}{1-(1-\frac{\Delta-1}{\Delta |Y|})^2} 
         \leq  \frac{1}{\frac{8}{|Y|} - \frac{8}{\Delta|Y|}-\frac{4}{|Y|^2}} 
         = \pbra{\frac{1}{8} + o_{\Delta}(1)} |Y|\mper
    \end{align*}
    That is, $V \leq \pbra{\frac{1}{8}+o_\Delta(1)}\abs{Y}$ with probability 1. Thus, combining it with the first conclusion of \Cref{lem:martingalediff} allows us to use Freedman's inequality to yield
    \begin{align*}
        \Pr\sbra{\abs{\abs{N\sbra{\tau}}- \E\sbra{|N\sbra{\tau}|}} > t}
        & \leq 2\exp\pbra{ - \frac{1}{2} \cdot \frac{t^2 }{\pbra{\frac{1}{8} + o_\Delta(1)}|Y| + t/3}} \mper
    \end{align*}
as desired.
\end{proof}

\begin{proof}[Proof of \Cref{lem:martingalediff}]
    Fix exposure results $\{\pi_j\}_{j = 1}^{i-1}$. For $(y,k) \in C(Y) \setminus \{\pi_j\}_{j = 1}^{i-1}$, we define 
    \begin{align*}
        f((y,k)) \coloneqq \E \sbra{|N\sbra{\tau}| \mid \pi_i = (y,k)}\mper
    \end{align*}
    Then, the two quantities that we want to bound are exactly $\max_{(y,k)\in C(Y) \setminus \{\pi_j\}_{j=1}^{i-1}} |f((y,k)) - \E\sbra{f}|$ and $\Var\sbra{f}$.
    
    Our first observation is that $|\text{image}(f)|\leq 2$, that is, it only takes two values which depend on whether its input $\frac{1}{\Delta}$-vertex $(y,k)$ has that the cloud $C(y)$ is completely unmatched thus far. Formally, let 
    \begin{align*}
        N \coloneqq \{(y,k) \mid C(y) \cap \{\pi_j\}_{j = 1}^{i-1}=\varnothing\}\mcom \;\; O \coloneqq C(Y)\setminus N\setminus \{\pi_j\}_{j = 1}^{i-1}\mper
    \end{align*}
    Informally, $N$ is the set of $\frac{1}{\Delta}$-vertices such that none of the elements of their clouds have appeared in $\pi_1,\dots,\pi_{i-1}$ and $O$ are the remaining unmatched $\frac{1}{\Delta}$-vertices.
    Observe that $f$ is {\em constant} in $N$ and $O$  by symmetry. We will show that for any  $(y_1,k_1) \in N$, $(y_2,k_2) \in O$, we have 
    \begin{equation}\label{eq:Zboundgoal1}|f((y_1,k_1)) - f((y_2,k_2))| \leq \pbra{1-\frac{\Delta-1}{\Delta|Y|}}^{\Delta|\tau|-i}\mper
    \end{equation}
    Since this quantity is less than $1$, it immediately implies the first conclusion.
    To see the second conclusion, suppose $|\text{image}(f)| = \{a ,b\}$; observe that $\Var\sbra{f}=p(1 - p) (a - b)^2$ where $a=f((y_1,k_1))$ and $b=f((y_2,k_2))$.
Then by \eqref{eq:Zboundgoal1}
$$\Var\sbra{f} \underset{p(1-p) \leq 1/4}{\leq} \frac{1}{4}(f((y_1,k_1)) - f((y_2,k_2)))^2 \leq \frac{1}{4} \pbra{1 - \frac{\Delta - 1}{\Delta |Y|}}^{2(\Delta|\tau| - i)} \mcom $$
as desired.
It remains to prove \eqref{eq:Zboundgoal1}. To start, we define $A, B: [\Delta|\tau|]\to  C(Y)$ to be uniformly random  one-to-one mappings such that 
    \begin{itemize}
        \item For any $1\leq j\leq i-1$, $A_j=B_j=\pi_j$, and
        \item $A_i=(y_1,k_1), B_i=(y_2,k_2)$.
    \end{itemize}

Let $N(A)$ and $N(B)$ be the number of neighbors of $\tau$ in the corresponding mappings.
    $$ f((y_1,k_1)) - f((y_2,k_2)) = \E\sbra{|N(A)|} - \E\sbra{|N(B)|}\mcom$$

    We define a joint probability distribution $\nu$ on $(A,B)$ as follows: for a given $A$ chosen uniformly at random, we let
    $$ B_j=\begin{cases}
        (y_2,k_2) & \text{if } j=i\\
        (y_1,k_1) & \text{if } A_j=(y_2,k_2)\\
        A_j &\text{otherwise}
        \end{cases}
    $$
    It is simple exercise that $B$ is distributed uniformly among all one-to-one mappings defined above.

    Define the event $$\mathcal{F} \coloneqq \{(A,B) \mid A_j \neq (y_2,k_2) \text{ for all } j \in [\Delta\abs{\tau}] \} \mper$$ 
    Observe that if $\mathcal{F}$ does not occur, then the image of $A,B$ are exactly the same so,  $\abs{N(A)} = \abs{N(B)}$. 

    On the other hand, if $\mathcal{F}$ occurs, then if $\{A_i\}_{i+1\leq j\leq\Delta|\tau|}$ all avoid $y_1$, then  $\abs{N(A)} = \abs{N(B)} + 1$, and otherwise $\abs{N(A)}=\abs{N(B)}$. So, we write
    \begin{align*}
        f((y_1,k_1)) - f((y_2,k_2)) & = \E_{(A,B)\sim\nu}\sbra{|N(A)|-|N(B)|} \\
        &\leq \E_{(A,B)\sim\nu} \sbra{|N(A)| - |N(B)| \mid \mathcal{F} } \\
        & \leq \prod_{j = i+1}^{\Delta|\tau|} \pbra{1 - \frac{\Delta - 1}{\Delta |Y| - j}} \leq \pbra{1 - \frac{\Delta - 1}{\Delta |Y|}}^{\Delta|\tau| - i}\mper
    \end{align*}
This proves \eqref{eq:Zboundgoal1} as desired. 
\end{proof}

\begin{proposition}
        \label{lem:singlebipartite}
        For any $\gamma < 1$, the bipartite pairing model on $(X,Y)$, and $\tau \subseteq X$ with $\abs{\tau} = \frac{\log \Delta}{(2+\gamma) \Delta} \abs{X}$ we have
        \begin{align*}
            (1 - o_\Delta (1))\frac{1}{\Delta^{1/(2+\gamma)}}\leq \Pr\sbra{y \notin N\sbra{\tau}} \leq \frac{1}{\Delta^{1/(2+\gamma)}}, \quad \forall y\in Y\mper
        \end{align*}
    \end{proposition}

    \begin{proof}
        We consider the probability that all $\frac{1}{\Delta}$-vertices in $C(y)$ avoid matching into $\tau$ sequentially. The probability that $(y,i)$ avoids $\tau$, conditioned on all previous copies avoiding $\tau$, is exactly $1 - \frac{\Delta \abs{\tau}}{\Delta \abs{X} - i}$. Therefore, the probability is given by
$$\Pr\sbra{y \notin N\sbra{\tau}} = \prod_{i=0}^{\Delta-1} \pbra{1 - \frac{\Delta \abs{\tau}}{\Delta \abs{X} - i}} \leq \left(1-\frac{\abs{\tau}}{|X|}\right)^\Delta \leq e^{-\Delta\abs{\tau}/|X|} = \Delta^{-1/(2+\gamma)}\mper $$
Similarly, 
\begin{align*}
    \Pr\sbra{y \notin N\sbra{\tau}} \geq \left(1-\frac{\abs{\tau}}{|X|-1}\right)^{\Delta} \underset{1-x\geq e^{-x-x^2/2}}{\geq} e^{-\frac{\Delta\abs{\tau}}{|X|-1} - \frac{\Delta\abs{\tau}^2}{2(|X|-1)^2}}  \geq \frac{1-o_\Delta(1)}{\Delta^{1/(2+\gamma)}}\mper
\end{align*}

    as desired.
\end{proof}

We now prove a slightly tighter bound for the expansion to be used when $\abs{\tau}$ is well above the threshold $\frac{\log \Delta}{(2+o_\Delta(1))\Delta}$ seen in \Cref{lem:bipartiteconc}.

\begin{lemma}
    \label{lem:bipartitelargeconc}
         For all $0 < a,b < 1$, there exists large enough $\Delta$ such that the following holds:
         for any random $\Delta$-regular bipartite $G = (X,Y,E)$, with high probability over the randomness of $G$, every $\tau \subseteq X$ of size $\abs{\tau} \geq \frac{|X|}{\Delta^a}$ satisfies $\abs{Y \setminus N\sbra{\tau}} < \frac{1}{\Delta^b} \abs{Y}\mper$
\end{lemma}
\begin{proof} 
Let $\tau \subseteq X, T \subseteq Y$. Similar to the proof of \Cref{lem:singlebipartite}, the probability that $N\sbra{\tau}$ fails to intersect $T$ is at most $\exp\pbra{ -\frac{\Delta |\tau| |T|}{|Y|}}$. Taking a union bound over all sets, 
    \begin{align*}
        \binom{|X|}{|\tau|} \binom{|Y|}{|T|} \exp\pbra{ -\frac{\Delta |\tau| |T|}{|Y|}} & \leq \exp\pbra{ |\tau| \log (e |X|/ |\tau|) + |T| \log (e|Y|/|T|) - \frac{\Delta |\tau| |T|}{|Y|}  } \\
        &\leq \exp\pbra{|X|( \Delta^{-a} O(\log\Delta) + \Delta^{-b} O(\log\Delta) - \Delta^{1 - a - b}) }\\
        &\leq \exp(O_\Delta(\abs{X}))\mper
    \end{align*} 
where the inequality holds by letting $|\tau| = \frac{|X|}{\Delta^a}$ and $|T|= \frac{|Y|}{\Delta^b}$ and using $|X|=|Y|$. 
Since $a,b<1$, for $\Delta$ large enough with high probability  every set $T$ and $N[\tau]$ intersects.
\end{proof}

\subsection{Bounds on the size of common neighborhoods}

We prove here a pair of results bounding the number of vertices sharing common neighbors.

\begin{lemma}
    \label{lem:neighborconc}
    Let $G = (X, Y, E)$ be a random $\Delta$-regular bipartite graph. Then with high probability over the randomness of $G$, every vertex $u, v \in X, Y$ satisfies $|N(u) \cap N(v)| \leq 2$.
\end{lemma}

    \begin{proof}

We show here an argument for subsets of $X$, the argument for subsets of $Y$ follow similarly and taking a union bound yields the full bound. Fix $x_1, x_2 \in X$ to be $u$ and $v$, which is without loss of generality since if one were from $Y$ they would not share neighbors. We argue $\Pr\sbra{\abs{N(x_1) \cap N(x_2)} > 2} \leq \frac{\poly(\Delta)}{\abs{X}^3}$ in the pairing model, after which the argument follows by union bound over the $2\binom{\abs{X}}{2} \leq \abs{X}^2$ such pairs and contiguity.

Let $\pi_{1, 1}, \dots, \pi_{1, \Delta}$ denote the matches for the cloud $C(x_1)$ and $\pi_{2, 1}, \dots, \pi_{2, \Delta}$ for $C(x_2)$ respectively. Within the pairing model, $\abs{N(x_1)\cap N(x_2)} > 2$ requires some three distinct pairs $(\pi_{1, \cdot}, \pi_{2, \cdot})$ where each pair is mapped to the same cloud in $Y$. Since there are only $\poly(\Delta)$ such pairs, it suffices to bound the probability $(\pi_{1,1}, \pi_{2,1}), (\pi_{1,1}, \pi_{2,1})$, and $(\pi_{1,1}, \pi_{2,1})$ all match by $\frac{\poly(\Delta)}{\abs{X}^3}$ and apply a union bound.

To finish, note that after conditioning on the matches $\pi_{1,1}, \pi_{1,2}$, and $\pi_{1,3}$, the matches for $\pi_{2,1}, \pi_{2,2}$, and $\pi_{2,3}$ are then a uniform subset of the remaining $\Delta\abs{Y}-3$ $\frac{1}{\Delta}$-vertices. Consider sampling the match for $\pi_{2,1}$. Crudely, the probability it collides with $\pi_{1,1}$ is at most $\frac{\Delta-1}{\Delta\abs{X}-3}$. Once conditioned on this, the probability $\pi_{2,2}$ collides with $\pi_{1, 2}$ with probability at most $\frac{\Delta-1}{\Delta\abs{X}-4}$. Similarly, for $\pi_{2,3}$ we get $\frac{\Delta-1}{\Delta\abs{X}-5}$. Simply multiplying achieves a bound of $\frac{\poly(\Delta)}{\abs{X}^3}$ as desired.
\end{proof}

\begin{lemma}
    \label{lem:neighborbound}
    Let $G = (X, Y, E)$ be a random $\Delta$-regular bipartite graph. Then with high probability over the randomness of $G$, there is no vertex $x \in X$ sharing at least 2 common neighbors with more than one other vertex.
\end{lemma}

\begin{proof}
    Fix $x \in X$ and let $\ell_x$ count the number of vertices in $X$ sharing at least 2 neighbors with $x$. Fix $N(x)$ and for a pair $(y_1, y_2) \in \binom{N(x)}{2}$, let $I^{(x)}_{y_1, y_2}$ be the indicator that $y_1$ and $y_2$ share a neighbor other than $x$. For any two pairs $(y_1, y_2), (y_3, y_4) \in \binom{Y}{2}$, let $I^{(x)}_{(y_1, y_2), (y_3, y_4)}$ be the indicator for $I^{(x)}_{y_1, y_2} \wedge I^{(x)}_{y_3, y_4}$. By conditioning on no vertices sharing more than $2$ common neighbors as in \Cref{lem:neighborconc}, we can then bound
    \begin{equation}\label{eq:commonneighbor}
        \ell_x \leq 1+ \sum_{(y_1, y_2) \neq (y_3, y_4) \in \binom{N(x)}{2}} I^{(x)}_{(y_1, y_2), (y_3, y_4)}\mper
    \end{equation}
    Our goal is then to show for any particular $y_1, y_2, y_3, y_4 \in \binom{N(x)}{2}$,
    \begin{align*}
        \Pr\sbra{I^{(x)}_{(y_1, y_2), (y_3, y_4)} > 0} \leq \frac{\poly(\Delta)}{\abs{X}^2}\mper
    \end{align*}
    Given \eqref{eq:commonneighbor}, we can simply union bound over every choice of $y_1, y_2, y_3, y_4$, of which there are at most $\poly(\Delta)$, and then union bound over every choice $x \in X$ to get a high probability bound in $\abs{X}$ that $\ell_x \leq 1$ for all of $X$.

    To prove the statement then, we work in the pairing model. Note it suffices to simply show $\Pr\sbra{I^{(x)}_{y_1, y_2} = 1} \leq \frac{\poly(\Delta)}{\abs{X}}$, since all $I^{(x)}_{y_1, y_2}$ and $I^{(x)}_{y_3, y_4}$ are pairwise independent when $(y_1, y_2) \neq (y_3, y_4)$.
    
    Now, let $\pi_{i, 1}, \dots, \pi_{i, \Delta}$ denote the $\Delta$ $\frac{1}{\Delta}$-vertex matches for $y_i$ within the model, and assume without loss of generality $\pi_{i, 1} \in C(x)$. $I^{(x)}_{y_1, y_2}$ then requires two of these matches (one from $y_1$ and one from $y_2$) fall in the same cloud in $X$. We begin by fixing $\pi_{1, 2}, \dots, \pi_{1, \Delta}$ and then sampling $\pi_{2, 2}, \dots, \pi_{2, \Delta}$ uniformly from what remains. By a simple union bound, it suffices to show just $\pi_{2,2}$ does not fall in the same cloud as any match for $y_1$ with high probability. To see this, note there are at most $\Delta$ unique clouds to avoid, each with $\Delta$ $\frac{1}{\Delta}$-vertices. The probability can then be observed to be at most $\frac{\Delta}{\abs{X}-\Delta}$, accounting for the $\frac{1}{\Delta}$-vertices already matched.
\end{proof}

\subsection{Bounds on the average size of an independent set}

The following claim bounds the size of a typical independent set from the hardcore model at fugacity $\lambda > 0$.

\begin{lemma}
    Let $G = (V, E)$ be a graph and $\mu$ the corresponding hardcore model at fugacity $\lambda > 0$. For any vertex $v \in V$, it holds that
    \begin{align*}
        \Pr_{I \sim \mu}\sbra{v \in I} \leq \frac{\lambda}{1+\lambda}\mper
    \end{align*}
\end{lemma}

\begin{proof}
From the definition of the hardcore distribution $\mu$ we have:
\begin{align*}
    \Pr_{I \sim \mu}\sbra{v \in I} &= \frac{\sum_{I \in \calI(G): v \in I} \lambda^{\abs{I}}}{\sum_{I \in \calI(G):v \in I} \lambda^{\abs{I}} + \sum_{I \in \calI(G):v \notin I} \lambda^{\abs{I}}}\\
    &= \frac{ \lambda \cdot \sum_{I: I\cup \{v\} \in \calI(G), v \notin I} \lambda^{\abs{I}}}{\lambda \cdot \sum_{I: I\cup\{v\} \in \calI(G), v \notin I} \lambda^{\abs{I}} + \sum_{I \in \calI(G): v \notin I} \lambda^{\abs{I}}}\\
    &\leq \frac{ \lambda \cdot \sum_{I: I\cup\{v\}\in \calI(G), v \notin I} \lambda^{\abs{I}}}{\lambda \cdot \sum_{I:I\cup\{v\} \in \calI(G),v \notin I} \lambda^{\abs{I}} + \sum_{I: I\cup\{v\} \in \calI(G), v \notin I} \lambda^{\abs{I}}} = \frac{\lambda}{\lambda+1}\mper
\end{align*}
\end{proof}

\begin{corollary}[\Cref{lem:indsetbound} restated]
    Let $G = (V, E)$ be a graph and $\mu$ the corresponding hardcore model at fugacity $\lambda > 0$. Then,
    \begin{align*}
        \Pr_{I \sim \mu}\sbra{\abs{I} \geq 4\lambda \abs{V}} \leq \pbra{\frac{e}{4}}^{4\lambda\abs{V}}\mper
    \end{align*}
\end{corollary}

\begin{proof}
    Fix an arbitrary set $S \in \binom{V}{4\lambda \abs{V} }$. By the previous lemma and sequential conditioning, we have the probability $S \subseteq I$ with $I \sim \mu$ is at most $\pbra{\frac{\lambda}{1+\lambda}}^{4\lambda \abs{V}}$. It then suffices to union bound over the entirety of $\binom{V}{4 \lambda \abs{V}}$. By standard binomial estimates we get
    \begin{align*}
        \binom{\abs{V}}{4\lambda\abs{V}} \cdot \pbra{\frac{\lambda}{1+\lambda}}^{4\lambda\abs{V}} \leq \pbra{\frac{e}{4\lambda}}^{4\lambda\abs{V}} \cdot \pbra{\frac{\lambda}{1+\lambda}}^{4\lambda\abs{V}} \leq \pbra{\frac{e}{4}}^{4\lambda\abs{V}}\mcom
    \end{align*}
    as desired.
\end{proof}

We also show that on a random $\Delta$-regular bipartite graph, the number of bipartite sets with large sizes in both sets is small.

\begin{lemma}[\Cref{prop:bipartiteindsetbound} restated]
        Let $G = (X, Y, E)$ be a random $\Delta$-regular bipartite graph and $\mu$ the corresponding hardcore model at fugacity $\lambda \lesssim \frac{1}{\sqrt{\Delta}}$. Then with high probability over the randomness in $G$,
        \begin{align*}
            \Pr_{I \sim \mu}\sbra{\abs{I \cap X} > \alpha\abs{X}, \abs{I \cap Y} > \alpha \abs{Y} } \leq \exp(-(\abs{X} + \abs{Y}))\mcom
        \end{align*}
        when $\alpha \leq \frac{\log \Delta}{(2+\gamma) \Delta}$ for $\gamma$ in \Cref{lem:bipartiteconc}.
    \end{lemma}

\begin{proof}
Throughout the proof, we assume \Cref{lem:bipartiteconc} holds for the graph, which holds with high probability over random $\Delta$-regular graphs. Let $I$ be a random sample of $\mu$ conditioned that $I_X = I \cap X$ satisfies $\abs{I_X} \geq \frac{\log \Delta}{(2+\gamma) \Delta} |X|$. By the lemma, we have that $\abs{Y \setminus N\sbra{I_X}} \leq  \frac{(\ell+1)\log \Delta}{\sqrt{\Delta}}$ for any choice $\ell > 0$.

Conditioned on $I_X$, the probability of a fixed subset of size $\frac{\log \Delta}{(2+\gamma) \Delta} |Y|$ is a subset of a sample from $\mu$ is exactly $\pbra{ \frac{\lambda}{1 + \lambda}}^{\frac{\log \Delta}{(2+\gamma) \Delta} |Y|}$. Then, union bounding over all subsets of size $\abs{Y \setminus N\sbra{I_X}}$ and using standard binomial estimates yields that the probability that any subset of size $\frac{\log \Delta}{(2+\gamma) \Delta} |Y|$ is sampled is at most 
\begin{align*}
    \binom{\abs{Y \setminus N\sbra{I_X}}}{\frac{\log \Delta}{(2+\gamma) \Delta} |Y|} \pbra{ \frac{\lambda}{1 + \lambda}}^{\frac{\log \Delta}{(2+\gamma) \Delta} |Y|} \leq \left((\ell+1)e(2+\gamma) \sqrt{\Delta} \cdot \lambda \right)^{\frac{\log \Delta}{(2+\gamma) \Delta} |Y|} \leq \exp\pbra{- \frac{|Y|}{\Delta}}\mper
\end{align*}
where the last inequality holds for large enough $\Delta$ and $\lambda \lesssim \frac{1}{\sqrt{\Delta}}$ with a sufficiently small constant. This means with such probability, none of the events occur over fixed $I_X$, which implies
\begin{align*}
        \Pr_{I \sim \mu}\sbra{\abs{I \cap X} > \frac{\log \Delta}{(2+\gamma) \Delta} |X|, \abs{I \cap Y} > \frac{\log \Delta}{(2+\gamma) \Delta} |Y|} & \leq \exp\pbra{- \frac{|Y|}{\Delta}} \mper
\end{align*}

as desired.
\end{proof}

\section{Trickle-down in the Two-sided Independent Set Slice}
\label{sec:bipartite1}

In this section, we show our first main result, \Cref{thm:bipartite}, which says that on bounded-degree random regular \textit{bipartite} graphs the down-up walk on the two-sided slice mixes in polynomial-time beyond the critical occupancy. The bulk of our proof goes into showing the following general bound on the top-link eigenvalues.

\begin{lemma}[Top-Link Expansion in the Two-Sided Independent Set Slice]
    \label{lem:main}
    Fix $k_X, k_Y \leq \abs{X}$ and let $(\calX, \mu^{(k_X, k_Y)})$ be the two-sided independent set slice of a $\Delta$-regular bipartite graph $G = (X,Y,E)$. Let $\tau = (\tau_X, \tau_Y) \in \calX$ satisfy $\codim(\tau) = 2$. We show the complex satisfies:
    \begin{enumerate}
        \item For $\tau \in \calI_{k_X-1, k_Y-1}(G)$, $\lambda_2(\P_\tau) \leq \frac{\lambda_2(\A_G)}{\min\cbra{\abs{X \setminus (\tau_X \cup N\sbra{\tau_Y})}-\Delta, \abs{Y \setminus (\tau_Y \cup N\sbra{\tau_X})} -\Delta } }$.
        \item For $\tau \in \calI_{k_X-2, k_Y}(G), \calI_{k_X, k_Y-2}(G)$, $\lambda_2(\P_\tau) \leq 0$.
    \end{enumerate}
\end{lemma}

We now show the proof of \Cref{thm:bipartite} assuming \Cref{lem:main}.

\begin{proof}[Proof of \Cref{thm:bipartite}]
    Let $G = (X,Y,E)$ be a random $\Delta$-regular bipartite graph and fix a pair of sizes $k_X, k_Y \leq \frac{\log \Delta}{(2+\gamma)\Delta} \abs{X}$ with $\gamma$ taken from \Cref{lem:bipartiteconc}, then consider the two-sided independent set slice $(\calX, \mu^{(k_X, k_Y)})$. Our goal is to show for any $\tau \in (\calX, \mu^{(k_X, k_Y)})$ with $\codim(\tau) = 2$, $\lambda_2(\P_\tau) \leq \frac{1}{2(\abs{\tau})}$, and further the complex is connected. 
    
    We begin with an application of \Cref{lem:bipartiteconc}, which says for any such $\tau$ with $\tau = (\tau_X, \tau_Y)$ and any $\ell > 0$ we have
    \begin{align*}
        \abs{X \setminus (\tau_X \cup N\sbra{\tau_Y})} \geq \frac{\ell\log \Delta}{\sqrt{\Delta}}\abs{X} - \frac{\log \Delta}{(2+\gamma)\Delta}\abs{X} \geq \frac{(\ell-1)\log \Delta}{\sqrt{\Delta}}\abs{X}\mcom
    \end{align*}
    and the analogous statement for $Y \setminus (\tau_Y \cup N\sbra{\tau_X})$. Note that the statement holds for sets $\tau$ of size $k_x$, but a link will have even fewer neighbors. We can then invoke \Cref{lem:main} for this complex which guarantees
    \begin{align*}
        \lambda_2(\P_\tau) \leq \frac{1}{\ell-1} \cdot \frac{\lambda_2(\A_G) \sqrt{\Delta}}{\log\Delta \abs{X}} \leq\frac{(2+\gamma)}{2(\ell-1)} \cdot \frac{\lambda_2(\A_G) \sqrt{\Delta}}{\log\Delta \abs{X}}\mper
    \end{align*}
    Applying \Cref{thm:friedman2}, we get with high probability $\lambda_2(\A_G) \leq 2\sqrt{\Delta-1} + o(1)$. Applying this yields
    \begin{align*}
        \lambda_2(\P_\tau) \leq \frac{1}{\ell-1} \cdot \frac{(2+\gamma)\Delta}{\log\Delta \abs{X}} \leq  \frac{1}{(\ell-1)\abs{\tau}}\mper
    \end{align*}
    The result follows by setting $\ell$ large enough.
    
    On the other hand, for any  $k_X, k_Y \leq \frac{\log\Delta}{(2+\gamma) \Delta} \abs{X}$ and the two-sided independent set slice $(\calX, \mu^{(k_X, k_Y)})$  of a random $\Delta$-regular bipartite graph $G = (X,Y,E)$, with high probability over the randomness of $G$, the 1-skeleton of any link of codimension at least $2$ is connected. This is because by \Cref{lem:main}, for any link $\tau$ of any size, if we truncate the complex to dimension $\abs{\tau}+2$, we have  $\lambda_2(\P_\tau) < 1$. Therefore, $(\calX, \mu^{(k_X,k_Y)})$ is connected.
\end{proof}

The rest of this section is then dedicated to proving the top-link expansion for the two-sided independent set slice seen in \Cref{lem:main}.

\begin{proof}[Proof of \Cref{lem:main}]

    While there are two types of $\tau \in \calX$ here, it is a simple exercise that if $\tau \in \calI_{k_X-2, k_Y}(G), \calI_{k_X, k_Y-2}(G)$ then the link is a complete graph so $\lambda_2(\P_\tau) \leq 0$. As such, we fix $\tau \in \calI_{k_X-1, k_Y-1}(G)$, and we prove the first bullet point. We start by observing the following form for the local walk matrix $\P_\tau$.

    \begin{observation}
        Let $(\calX, \mu^{(k_X, k_Y)})$ be the two-sided independent set slice of a bipartite graph $G = (X, Y, E)$ and let $\tau \in \calI_{k_X-1, k_Y-1}(G)$. Define the graph $G_\tau = \overline{G\sbra{(X \cup Y) \setminus (\tau \cup N\sbra{\tau})}}$ where here $\overline{G}$ is the bipartite complement with respect to $(X, Y)$. Then we have
    \begin{align*}
        \P_\tau = \Gamma^{-1}_{G_\tau}\A_{G_\tau}\mcom
    \end{align*}
    up to empty rows and columns. Here $\A_{G_\tau}$ is the adjacency matrix of $G_\tau$ and $\Gamma_{G_\tau}$ is its diagonal degree matrix.
    \end{observation}

This characterization looks exactly like \Cref{prop:regularwalk} (which characterizes the links of the regular independent set slice) with the complement replaced with the bipartite complement. This is because for $\tau \in \calI_{k_X-1, k_Y-1}(G)$ we are no longer allowed to take two vertices from the same side, so even though there are no edges between vertices on the same side they are not connected in the link.

We can then use this characterization directly to try and derive a second eigenvalue bound.
We will show that
\begin{equation}\label{eq:bipfixedgoal1}
\lambda_2(\A_{G_\tau})\leq \lambda_2(\A_G)\mper
\end{equation}
Assuming this, 
we observe conjugating by $\Gamma_{G_\tau}^{1/2}$ yields
    \begin{align*}
        \lambda_2(\P_\tau) = \lambda_2(\Gamma_{G_\tau}^{1/2} \P_\tau \Gamma_{G_\tau}^{-1/2}) = \lambda_2(\Gamma_{G_\tau}^{-1/2} \A_{G_\tau} \Gamma_{G_\tau}^{-1/2}) \leq \lambda_2(\A_G) \cdot \norm{\Gamma_{G_\tau}^{-1}}{2}\mper
    \end{align*}
    We conclude the proof by observing $\norm{\Gamma_{G_\tau}^{-1}}{2}$ is just the reciprocal of the smallest degree in $G_\tau$. Since $G_\tau$ is the complement, this is always at least the smaller of $\frac{1}{\abs{X \setminus (\tau_X \cup N\sbra{\tau_Y})}-\Delta}$ or $\frac{1}{\abs{Y \setminus (\tau_Y \cup N\sbra{\tau_X})}-\Delta}$. It remains to prove \eqref{eq:bipfixedgoal1}.

    Let $\sigma_X = X \setminus (\tau_X \cup N\sbra{\tau_Y})$ and define $\sigma_Y$ symmetrically. The key observation is that we can rewrite
    \begin{align*}
        \A_{G_\tau} = \mathbf{1}_{\sigma_X}\mathbf{1}_{\sigma_Y}^\top + \mathbf{1}_{\sigma_X}\mathbf{1}_{\sigma_Y}^\top - \A_{G\sbra{V \setminus (\tau \cup N\sbra{\tau})}}\mper
    \end{align*}
    This is the exact definition of the bipartite complement, flipping every entry crossing $X$ and $Y$ and noting the remainder are $0$ throughout. Notice this is a principal submatrix of the entire bipartite complement $\overline{\A_G}$, so by the \nameref{fact:cauchyinterlacing} we have
    \begin{align*}
        \lambda_2(\A_{G_\tau}) \leq \lambda_2(\overline{\A_G})\mper
    \end{align*}
    So, to prove \eqref{eq:bipfixedgoal1} it is enough to show that 
    \begin{equation}
        \label{eq:bipfixgoal2}
        \lambda_2(\overline{\A_G}) \leq \lambda_2(\A_G)\mper
    \end{equation}  
    We have an explicit form for $\overline{\A_G}$ as
    \begin{align*}
        \A_G = \begin{bmatrix}
            0 & A\\
            A^\top & 0
        \end{bmatrix} \text{ and } \overline{\A_G} = \begin{bmatrix}
            0 & \mathbf{1}\mathbf{1}^\top - A\\
            (\mathbf{1}\mathbf{1}^\top - A)^\top & 0
        \end{bmatrix}\mper
    \end{align*}

    We observe that $\lambda_2(\overline{\A_G}) \leq \sqrt{\lambda_3(\overline{\A_G})^2}$ since $\overline{\A_G}$ being $\Delta$-biregular implies $\overline{\A_G}$ has two maximal eigenvalues $\Delta^2$ trivially. We can compute in block matrix form
    \begin{align*}
        \overline{\A_G}^2 &= \begin{bmatrix}
            (\mathbf{1}\mathbf{1}^\top - A)^2 & 0\\
            0 & (\mathbf{1}\mathbf{1}^\top - A)^2
        \end{bmatrix}\\
        &=\begin{bmatrix}
            (\abs{X}+\abs{Y} - \Delta^2)\mathbf{1}\mathbf{1}^\top + A^2 & 0\\
            0 & (\abs{X}+\abs{Y} - \Delta^2)\mathbf{1}\mathbf{1}^\top + A^2
        \end{bmatrix} \\
        &= (\abs{X} + \abs{Y} - \Delta^2) (\mathbf{1}_X\mathbf{1}_X^\top + \mathbf{1}_Y\mathbf{1}_Y^\top) + \A_G^2\mper
    \end{align*}
    The last equality follows from the fact that $\mathbf{1}$ is a left and right eigenvector of $A$ with eigenvalue $\Delta$. Now we can simply observe that the $\A_G^2$ is the union of two $\Delta^2$-regular graphs and so its top eigenspace includes $\spn \cbra{\mathbf{1}_X, \mathbf{1}_Y}$. As a result 
    \begin{align*}
        \lambda_3(\overline{\A_G}^2) \leq \lambda_3(\A_G^2) = \lambda_2(\A_G)^2\mcom
    \end{align*}
    proving \eqref{eq:bipfixgoal2} as desired.
\end{proof}
\section{Trickle-down in the One-sided Slice of the Hardcore Model}
\label{sec:bipartite2}

In this section, we prove \Cref{thm:bipartite2}, which says that on random regular bipartite graphs the down-up walk on the one-sided slice mixes in polynomial-time at certain occupancies. Mirroring the previous section, our proof is centered around the following top-link eigenvalue bound.

\begin{lemma}[Top-Link Expansion in the One-Sided Slice of the Hardcore Model]
    \label{lem:secondeigen}
    Fix $k \leq \abs{X}$ and let $(\calX, \mu_\lambda^{(k)})$ be the one-sided slice of the hardcore model of a $\Delta$-regular bipartite graph $G = (X,Y,E)$ at fugacity $\lambda > 0$. Let $\tau \subseteq X$ satisfy $\codim(\tau) = 2$. Finally, assume $\max_{u,v \in \calX_\tau(1)} \abs{N_\tau(u) \cap N_\tau(v)} \leq 2$ and no vertex $u \in \calX_\tau(1)$ shares $2$ common neighbors with more than one other vertex in $G$. Then,
    \begin{align*}
        \lambda_2(\P_\tau) \leq \frac{(\lambda \cdot \lambda_2(\A_G)^2 + \lambda^2-1) \cdot (1+\lambda)^{\Delta_\tau}}{\abs{X} -\abs{\tau} - (1+\lambda)^{\Delta_\tau}}\mcom
    \end{align*}
    where $\Delta_\tau \coloneqq \E_{\substack{v \sim \calX_\tau(1)}}\sbra{\abs{N_\tau(v)}}$ is the average degree in the link of $\tau$.
\end{lemma}

Assuming this we prove \Cref{thm:bipartite2}.
\begin{proof}[Proof of \Cref{thm:bipartite2}]
    Let $G = (X,Y,E)$ be a random $\Delta$-regular bipartite graph and fix a size $\frac{\log \Delta}{(2+\gamma)\Delta}\abs{X}+2 \leq k \leq 4\lambda \abs{X}$ where $\gamma$ is taken from \Cref{lem:bipartiteconc}. We now consider the one-sided slice $(\calX, \mu_\lambda^{(k)})$ for $\lambda$ to be chosen later and fix $\tau \subseteq X$ with $\codim(\tau) = 2$. Our goal is to apply \Cref{lem:secondeigen} to show $\lambda_2(\P_\tau) < \frac{1}{2(\abs{\tau}+2)}$. To satisfy the required assumptions $\max_{u,v \in \calX_\tau(1)} \abs{N_\tau(u) \cap N_\tau(v)} \leq 2$ and no vertex $u \in \calX_\tau(1)$ shares $2$ common neighbors with more than one other vertex in $G$, we use \Cref{lem:neighborconc} and \Cref{lem:neighborbound}, which hold since $G$ is a random regular bipartite graph.

    The remaining step is then to bound $(1+\lambda)^{\Delta_\tau}$. We show
    \begin{equation}\label{eq:exp}
        (1+\lambda)^{\Delta_\tau} \leq \Delta^{a-\frac{1}{2}}\mcom
    \end{equation}
    for any $a \in (1/2, 1)$. This bound is enough if $|\tau|$ is relatively close to the critical occupancy threshold. 
    For larger sizes, say $\abs{\tau} \geq \frac{\abs{X}}{\Delta^a}$, we will show that
    \begin{equation}\label{eq:explarge}
        (1+\lambda)^{\Delta_\tau} \leq 2 \mper
    \end{equation}
    
    Assuming \eqref{eq:exp} and \eqref{eq:explarge}, we complete the proof. First, 
    \begin{align*}
        \lambda_2(\P_\tau) \leq \frac{(\lambda \cdot \lambda_2(\A_G)^2 + \lambda^2-1) \cdot (1+\lambda)^{\Delta_\tau}}{\abs{X} -\abs{\tau} - (1+\lambda)^{\Delta_\tau}} \underset{\eqref{eq:exp}}{\leq} \frac{(\lambda \cdot \lambda_2(\A_G)^2 + \lambda^2-1) \cdot \Delta^{a-\frac{1}{2}}}{\abs{X} -\abs{\tau} - \Delta^{a-\frac{1}{2}}}\mper
    \end{align*}
    Next, we invoke \Cref{thm:friedman2} for $G$, which holds with high probability to conclude $\lambda_2(\A_G) \leq 2\sqrt{\Delta-1} + o(1)$. Applying this along with $\lambda < 1$ to the above gives us
    \begin{align*}
        \lambda_2(\P_\tau) \leq \frac{4\lambda \Delta \cdot \Delta^{a-\frac{1}{2}}}{\abs{X}-\abs{\tau}-\Delta^{a-\frac{1}{2}}} \underset{\abs{\tau} \leq \frac{\abs{X}}{\Delta^a}}{=} (4+o_\Delta(1)) \frac{\lambda \Delta^{\frac{1}{2}+a}}{\abs{X}}\mper
    \end{align*}
    By setting $\lambda \lesssim \frac{1}{\sqrt{\Delta}}$ with a small enough constant this bound becomes $\frac{\Delta^a}{3\abs{X}}$, which can be made less than $\frac{1}{2(\abs{\tau}+2)}$ assuming $\abs{\tau} \leq \frac{\abs{X}}{\Delta^a}$.
    
    On the other hand, if  $\abs{\tau} \geq \frac{\abs{X}}{\Delta^a}$, then using \eqref{eq:explarge} in place of \eqref{eq:exp} we have
    \begin{align*}
        \lambda_2(\P_\tau) \leq \frac{2(\lambda \cdot \lambda_2(\A_G)^2 + \lambda^2-1) }{\abs{X} -\abs{\tau} -2} \underset{\lambda_2(\A_G) \leq 2\sqrt{\Delta-1}+o(1)}{\leq} \frac{8 \lambda \Delta}{\abs{X}-\abs{\tau}-2} \underset{\abs{\tau} \leq \frac{\abs{X}}{\sqrt{\Delta}}}{=} (8+o_\Delta(1)) \frac{\lambda \Delta}{\abs{X}}\mper
    \end{align*}
    Once again setting $\lambda \lesssim \frac{1}{\sqrt{\Delta}}$ with a small enough constant gives the bound $\frac{\sqrt{\Delta}}{3\abs{X}}$. Since $\abs{\tau} \leq 4\lambda \abs{X}$, we can assume $\abs{\tau} \leq \frac{\abs{X}}{\sqrt{\Delta}}$, which makes this bound good enough to achieve $\frac{1}{2(\abs{\tau}+2)}$. It suffices then to prove \eqref{eq:exp} and \eqref{eq:explarge} to finish.

Starting with \eqref{eq:exp}, we begin by bounding $\Delta_\tau = \E_{\substack{v \sim \calX_\tau(1)}}\sbra{\abs{N_\tau(v)}}$. To do so, we observe $\Delta_\tau$ is just the normalized count of the number of edges going from $\calX_\tau(1)$ to $Y \setminus N\sbra{\tau}$. This quantity is then exactly $\Delta \cdot \abs{Y \setminus N\sbra{\tau}}$, since none of the edges from $Y \setminus N\sbra{\tau}$ can go into $\calX_\tau(1)$ and $G$ is $\Delta$-regular. Finally, we invoke \Cref{lem:bipartiteconc} to get $\abs{Y \setminus N\sbra{\tau}} \leq \frac{(\ell+1) \log \Delta}{\sqrt{\Delta}}\abs{Y}$ for some constant $\ell > 0$ to be chosen later. This allows us to immediately bound
\begin{align*}
    \Delta_\tau \leq \frac{\Delta \cdot \frac{(\ell+1) \log \Delta}{\sqrt{\Delta}}\abs{Y}}{\abs{X}-\abs{\tau}} \leq (\ell+1+o_\Delta(1)) \cdot \sqrt{\Delta}\log\Delta\mper
\end{align*}
Substituting this into $(1+\lambda)^{\Delta_\tau}$ we observe
\begin{align*}
    (1 + \lambda)^{\Delta_\tau} \leq \exp\pbra{\lambda \cdot (\ell+1+o_\Delta(1)) \sqrt{\Delta}\log\Delta)} \leq \Delta^{a-\frac{1}{2}}\mcom
\end{align*}
where the exponent in the last term is made as small as needed by choosing $\lambda \lesssim \frac{1}{\sqrt{\Delta}}$.

Similarly, for \eqref{eq:explarge}, note since $a < 1$ we can use \Cref{lem:bipartitelargeconc} to achieve the improved bound $\abs{Y \setminus N\sbra{\tau}} \leq \frac{\abs{Y}}{\Delta^b}$ for any constant $b \in (0,1)$. This allows us to prove
\begin{align*}
    \Delta_\tau \leq \frac{\Delta \cdot \frac{\abs{Y}}{\Delta^b}}{\abs{X}-\abs{\tau}} \leq (1+o_\Delta(1)) \cdot \Delta^{1-b}\mcom
\end{align*}
and once again substitute to get
\begin{align*}
    (1 + \lambda)^{\Delta_\tau} \leq \exp \left(\lambda \cdot (1+o_\Delta(1)) \Delta^{1-b}\right) < 2\mcom
\end{align*}
by taking $b > \frac{1}{2}$ and $\lambda \lesssim \frac{1}{\sqrt{\Delta}}$.
\end{proof}

In the rest of this section we prove \Cref{lem:secondeigen}.

\begin{proof}[Proof of \Cref{lem:secondeigen}]

    We begin by introducing the following (multi)graph we call the $\tau$-neighbor graph.
    \begin{definition}[$\tau$-neighbor graph]
        \label{def:nbrgraph}
        Let $G = (X, Y, E)$ be a bipartite graph. For $\tau \subseteq X$, let $G_\tau$ be the graph on $X \cup Y$ with all edges adjacent to $\tau \cup N\sbra{\tau}$ removed. The ($X$-sided) $\tau$-neighbor graph of $G$ is then $H_\tau = (G_\tau^2)\sbra{\calX_\tau(1)}$. In other words, $H_\tau$ is the graph with adjacency matrix $\A_{H_\tau} \in \R^{\calX_\tau(1) \times \calX_\tau(1)}$ satisfying for $u, v \in \calX_\tau(1)$
        \begin{align*}
            \A_{H_\tau}(u,v) = \abs{N_\tau(u) \cap N_\tau(v)}\mper
        \end{align*}
    \end{definition}

    Our first key observation is that the $\tau$-neighbor graph characterizes the top-link random walk operators via the following equation.
    \begin{equation}\label{eq:bip2goal2}
        \Pi_\tau\P_\tau  = \frac{1}{2Z_\tau}\wt\Pi_\tau (1+\lambda)^{\A_{H_\tau}} \wt\Pi_\tau\mper
    \end{equation}
    We define the matrix exponential\footnote{Note that this differs than the more standard functional definition of $c^X$.} here as
    \begin{align*}
        (1+\lambda)^{\A_{H_\tau}} = \begin{cases}
            (1+\lambda)^{\abs{N_\tau(u) \cap N_\tau(v)}} & u \neq v \in \calX_\tau(1)\\
            0 & \text{otherwise}
        \end{cases}\mcom
    \end{align*}
    let $\wt\Pi_\tau = \Gamma_\tau^{-1}\Pi_\tau \text{ for } \Gamma_\tau(u,u) = \frac{Z_\tau^{(u)}}{2Z_\tau}$, and define the constants $Z_\tau$ and $Z^{(v)}_\tau$ as
    \begin{align*}
        &Z^{(u)}_\tau \coloneqq \sum_{\substack{v \in \calX_\tau(1) \\ v \neq u}} (1+\lambda)^{-\abs{N_\tau(v)}}(1+\lambda)^{\abs{N_\tau(u) \cap N_\tau(v)}}\mcom\\
        &Z_\tau \coloneqq \sum_{w \in \calX_\tau(1)} (1+\lambda)^{-\abs{N_\tau(w)}} \cdot Z_\tau^{(w)} \mper
    \end{align*}
    To show \eqref{eq:bip2goal2} it suffices to compute entry-wise $\P_\tau(u,v) = \frac{(1+\lambda)^{-\abs{N_\tau(v)}}(1+\lambda)^{\abs{N_\tau(u) \cap N_\tau(v)}}}{Z_\tau^{(u)}}$ and $\Pi_\tau(u,u) = \frac{1}{2Z_\tau} (1+\lambda)^{-\abs{N_\tau(u)}} \cdot Z_\tau^{(u)}$ directly from the definition of the one-sided slice.

    Our primary goal is then to show
    \begin{equation}\label{eq:bip2goal1}
        (1+\lambda)^{\A_{H_\tau}} \preceq (\mathbf{1}\mathbf{1}^\top + \lambda \A_{G^2\sbra{\calX_\tau(1)}} + (\lambda^2 - 1) \Id)\mper
    \end{equation}
    Combining this with \eqref{eq:bip2goal2} immediately yields
    \begin{equation*}
        \Pi_\tau \P_\tau \preceq \frac{1}{2Z_\tau}\wt\Pi_\tau (\mathbf{1}\mathbf{1}^\top + \lambda \A_{G^2\sbra{\calX_\tau(1)}} + (\lambda^2 - 1) \Id) \wt\Pi_\tau \mper
    \end{equation*}

    To see why this is sufficient, observe that if we subtract a certain multiple of $\wt\Pi_\tau \mathbf{1}\mathbf{1}^\top \wt\Pi_\tau$ from above we get
    \begin{align*}
         \Pi_\tau \P_\tau - \frac{1}{2Z_\tau}\left(1+\frac{\lambda \Delta^2}{\abs{X}}\right) \wt\Pi_\tau \mathbf{1}\mathbf{1}^\top \wt\Pi_\tau &\preceq \frac{1}{2Z_\tau} \wt\Pi_\tau \pbra{\lambda \A_{G^2\sbra{\calX_\tau(1)}} - \frac{\lambda \Delta^2}{\abs{X}} \mathbf{1}\mathbf{1}^\top + (\lambda^2-1)\Id} \wt\Pi_\tau\\
         &\preceq \frac{\lambda \cdot \lambda_2(\A_G)^2 + \lambda^2 - 1}{2Z_\tau} \wt\Pi_\tau^2\\
         &\preceq \max_{u \in \calX_\tau(1)} \frac{\lambda \cdot \lambda_2(\A_G)^2 + \lambda^2 - 1 }{2Z_\tau} \cdot 
         \frac{2Z_\tau}{Z_\tau^{(u)}}(1+\lambda)^{-\abs{N_\tau(u)}} \cdot \Pi_\tau\\
         &\preceq \max_{u \in \calX_\tau(1)} \frac{\lambda \cdot \lambda_2(\A_G)^2 + \lambda^2 - 1 }{Z_\tau^{(u)}} \Pi_\tau\mper
    \end{align*}

    The second line uses that $\lambdamax(\A_{G^2\sbra{\calX_\tau(1)}} - \frac{\Delta^2}{\abs{X}} \mathbf{1}\mathbf{1}^\top) \leq \lambda_2(\A_G)^2$. To see this, we first extend $\A_{G^2\sbra{\calX_\tau(1)}}$ to $\R^{X \times X}$ by appending $0$s. We can then apply the \nameref{fact:cauchyinterlacing} to see $\A_{G^2\sbra{\calX_\tau(1)}} \preceq \A_{G^2\sbra{X}}$, using implicitly that $\A_{G^2\sbra{X}}$ is a submatrix of $\A_G^2$ and is therefore positive semidefinite. We then use that $\mathbf{1}$ is the top eigenvector of $\A_{G^2\sbra{X}}$ with eigenvalue $\Delta^2$ since $G^2\sbra{X}$ is regular and conclude with $\lambda_2(\A_{G^2\sbra{X}}) \leq \lambda_2(\A_G)^2$. Finally, the fourth line follows since $\Gamma_\tau^{-1} \Pi_\tau \Gamma_\tau^{-1}(i,i) = \frac{2Z_\tau}{Z_\tau^{(i)}} (1+\lambda)^{-\abs{N_\tau(i)}}$.  From here, we can multiply both sides of each equation by $\Pi_\tau^{-1/2}$ and by the similarity of $\Pi_\tau^{1/2}\P_\tau \Pi_\tau^{-1/2}$ and $\P_\tau$ and the \nameref{fact:cauchyinterlacing} we conclude $\lambda_2(\P_\tau) \leq \max_{i \in \calX_\tau(1)} \frac{(\lambda \cdot \lambda_2(\A_G)^2 + \lambda^2 - 1)}{Z_\tau^{(i)}}$.

    We finish by applying Jensen's inequality to see:
    \begin{align*}
        Z_\tau^{(u)} &= \sum_{\substack{v \in \calX_\tau(1) \\ v \neq u}} (1+\lambda)^{-\abs{N_\tau(v)}}(1+\lambda)^{\abs{N_\tau(u) \cap N_\tau(v)}}\\
        &\geq \sum_{\substack{v \in \calX_\tau(1)}} (1+\lambda)^{-\abs{N_\tau(v)}} - (1+\lambda)^{-\abs{N_\tau(u)}}\\
        &\geq (\abs{X}-\abs{\tau}) \E_{\substack{v \sim \calX_\tau(1)}}\sbra{(1+\lambda)^{-\abs{N_\tau(v)}}} - 1\\
        &\geq \frac{ \abs{X} - \abs{\tau}}{(1+\lambda)^{\E_{\substack{v \sim \calX_\tau(1)}}\sbra{\abs{N_\tau(v)}}} } - 1\mcom
    \end{align*}
    where $v \sim \calX_\tau(1)$ is with respect to the uniform measure. Plugging this into the 
above immediately yields $\lambda_2(\P_\tau) \leq \frac{(\lambda \cdot \lambda_2(\A_G)^2 + \lambda^2-1) \cdot (1+\lambda)^{\E_{\substack{v \sim \calX_\tau(1)}}\sbra{\abs{N_\tau(j)}}}}{\abs{X} -\abs{\tau} - (1+\lambda)^{\E_{\substack{v \sim \calX_\tau(1)}}\sbra{\abs{N_\tau(j)}}}}$ as desired.

It remains to prove
\eqref{eq:bip2goal1}.     We do that  in two parts. First, assuming $\max_{u \neq v \in \calX_\tau(1)} \abs{N_\tau(u) \cap N_\tau(v)} \leq 2$ and that no $u \in \calX_\tau(1)$ shares $2$ common neighbors with more than one other vertex we have 
    \begin{equation}\label{eq:bip2goal3}
        (1+\lambda)^{\A_{H_\tau}} \preceq \mathbf{1}\mathbf{1}^\top + \lambda \A_{H_\tau} + (\lambda^2-1) \Id \mper
    \end{equation}
    Second, we show
    \begin{equation}\label{eq:bip2goal4}
        \A_{H_\tau} \preceq \A_{G^2\sbra{\calX_\tau(1)}}\mper
    \end{equation}
    Putting \eqref{eq:bip2goal3} and \eqref{eq:bip2goal4} together immediately gives \eqref{eq:bip2goal1}, so we prove them in turn.

    Towards \eqref{eq:bip2goal3}, let $\A^{(\ell)} \in \R^{\calX_\tau(1) \times \calX_\tau(1)}$ be defined as the submatrix of $\A$ consisting only of entries with value exactly $\ell$, the rest zeroed out. Let $\A'_{H_\tau} = \A_{H_\tau} - \diag(\A_{H_\tau})$. As such, we can write
        \begin{align*}
            (1+\lambda)^{\A_{H_\tau}} &= \mathbf{1}\mathbf{1}^\top - \Id + \lambda \A'^{(1)}_{H_\tau} + \pbra{\lambda + \frac{\lambda^2}{2}} \A'^{(2)}_{H_\tau}\\
            &= \mathbf{1}\mathbf{1}^\top - \Id + \lambda \A'_{H_\tau} + \frac{\lambda^2}{2} \A'^{(2)}_{H_\tau}\mper
        \end{align*}
        First, we observe $\A'_{H_\tau} \preceq \A_{H_\tau}$ since the diagonal is positive semi-definite in $\A_{H_\tau}$. It suffices then to bound the spectral norm of the latter term, and we have $\norm{\A'^{(2)}_{H_\tau}}{2} \leq 2$ which is just the maximum row sum. To see this, we invoke \Cref{lem:neighborbound}, which guarantees no vertex shares $2$ neighbors with more than one other vertex. This proves \eqref{eq:bip2goal3}.

    We finish by showing \eqref{eq:bip2goal4}. We write
        \begin{align*}
            &\A_G \id_{V \setminus (\tau \cup N\sbra{\tau})} \A_G \preceq \A_G^2\\
            &\implies \id_{\calX_\tau(1)} \A_G \id_{V \setminus (\tau \cup N\sbra{\tau})} \A_G \id_{\calX_\tau(1)} \preceq \id_{\calX_\tau(1)} \A_G^2 \id_{\calX_\tau(1)}\\
            &\implies \A_{H_\tau} \preceq \A_{G^2\sbra{\calX_\tau(1)}}\mper
        \end{align*}
         The second line uses that for $A \preceq B$ and $P \succeq 0$, $PAP^\top \preceq PBP^\top$. The third line then invokes the definitions of $H_\tau$ and $G^2\sbra{\calX_\tau(1)}$ directly.
    \end{proof}

\bibliographystyle{alpha}
\bibliography{references}

\appendix
\section{Trickle-down in the Independent Set Slice}
\label{sec:regular}

In this section, we give a simple proof of \Cref{thm:regular}, showing on bounded-degree random regular graphs the natural down-up walk on the slice mixes in polynomial-time beyond the critical occupancy. As in previous sections, we begin by bounding the second eigenvalue of the top-links.

\begin{lemma}[Top-Link Expansion in the Independent Set Slice]
    \label{lem:regular}
    Let $(\calX, \mu^{(k)})$ be the independent set slice of a $\Delta$-regular graph $G = (V, E)$ and let $\tau \in \calI_{k-2}(G)$. Then we have
    \begin{align*}
        \lambda_2(\P_\tau) \leq \frac{-\lambdamin(\A_G)-1}{\abs{V \setminus (\tau \cup N\sbra{\tau})}-\Delta}\mper
    \end{align*}
\end{lemma}

Note, a similar result appeared in \cite{AlevL20} as one of the first applications of the trickle-down machinery towards Markov chain mixing. Our proof is conceptually the same but we repeat it to highlight a particular optimization we make in order to improve the result for random regular graphs, but first we show how \Cref{lem:regular} can be used to prove \Cref{thm:regular}. We will additionally need \Cref{lem:regularconc}, the random regular analogue of \Cref{lem:bipartiteconc}, which we delay until the next section.

\begin{proof}[Proof of \Cref{thm:regular}]
    Let $G = (V,E)$ be a random $\Delta$-regular graph and fix a size $k \leq \frac{\log \Delta}{(2+\gamma)\Delta} \abs{V}+2$ where $\gamma$ is taken from \Cref{lem:regularconc}, then consider the independent set slice $(\calX, \mu^{(k)})$. Our goal is to show for any $\tau \in (\calX, \mu^{(k)})$ with $\codim(\tau) = 2$, $\lambda_2(\P_\tau) \leq \frac{1}{2(\abs{\tau}+2)}$, and further, the complex is connected. 
    
    We begin with an application of \Cref{lem:regularconc}, which says for any such $\tau$ and any $\ell > 0$ we have
    \begin{align*}
        \abs{V \setminus (\tau \cup N\sbra{\tau})} \geq \frac{\ell\log \Delta}{\sqrt{\Delta}}\abs{V}\mper
    \end{align*}
    We can then invoke \Cref{lem:regular} for this complex which guarantees
    \begin{align*}
        \lambda_2(\P_\tau) \leq \frac{1}{\ell} \cdot \frac{-\lambdamin(\A_G) \sqrt{\Delta}}{\log\Delta \abs{V}} \leq\frac{(2+\gamma)}{2\ell} \cdot \frac{-\lambdamin(\A_G) \sqrt{\Delta}}{\log\Delta \abs{V}}\mper
    \end{align*}
    Applying \Cref{thm:friedman}, we get with high probability $-\lambdamin(\A_G) \leq 2\sqrt{\Delta-1} + o(1)$. Applying this yields
    \begin{align*}
        \lambda_2(\P_\tau) \leq \frac{1}{\ell} \cdot \frac{(2+\gamma)\Delta}{\log\Delta \abs{V}} \leq  \frac{1}{\ell\abs{\tau}}\mper
    \end{align*}
    The result follows by setting $\ell$ large enough.
    
    On the other hand, for any  $k \leq \frac{\log\Delta}{(2+\gamma) \Delta} \abs{V}+2$ and the independent set slice $(\calX, \mu^{(k)})$  of a random $\Delta$-regular graph $G = (V,E)$, with high probability over the randomness of $G$, the 1-skeleton of any link of codimension at least $2$ is connected. This is simply because by \Cref{lem:regular}, for any link $\tau$, if we truncate the complex to dimension $\abs{\tau}+2$, we have  $\lambda_2(\P_\tau) < 1$. Therefore, $(\calX, \mu^{(k)})$ is connected.
\end{proof}

We devote the rest of this section to proving \Cref{lem:regular} and \Cref{lem:regularconc}.

\begin{proof}[Proof of \Cref{lem:regular}]
    We start with the following observation, characterizing the random walk matrix within each top-link.
    \begin{observation}
        \label{prop:regularwalk}
        Let $(\calX, \mu^{(k)})$ be the independent set slice of a graph $G = (V, E)$ and let $\tau \in \calI_{k-2}(G)$. Define the graph $G_\tau = \overline{G\sbra{V \setminus (\tau \cup N\sbra{\tau})}}$ where $\overline{G}$ is the graph complement. Then we have
        \begin{align*}
            \P_\tau = \Gamma^{-1}_{G_\tau}\A_{G_\tau}\mcom
        \end{align*}
        up to empty rows and columns. Here $\A_{G_\tau}$ is the adjacency matrix of $G_\tau$ and $\Gamma_{G_\tau}$ is its diagonal degree matrix.
    \end{observation}

    With this characterization in mind, we will show that
\begin{equation}\label{eq:fixedgoal1}
\lambda_2(\A_{G_\tau})\leq -\lambdamin(\A_G)-1\mper
\end{equation}
Assuming this, 
we observe conjugating by $\Gamma_{G_\tau}^{1/2}$ yields
    \begin{align*}
        \lambda_2(\P_\tau) = \lambda_2(\Gamma_{G_\tau}^{1/2} \P_\tau \Gamma_{G_\tau}^{-1/2}) = \lambda_2(\Gamma_{G_\tau}^{-1/2} \A_{G_\tau} \Gamma_{G_\tau}^{-1/2}) \leq -\lambdamin(\A_G) \cdot \norm{\Gamma_{G_\tau}^{-1}}{2}\mper
    \end{align*}
    We conclude the proof by observing $\norm{\Gamma_{G_\tau}^{-1}}{2}$ is just the reciprocal of the smallest degree in $G_\tau$. Since $G_\tau$ is the complement, this is always $\frac{1}{\abs{V \setminus (\tau \cup N\sbra{\tau})} - \Delta}$. It remains to prove \eqref{eq:fixedgoal1}.

    The key observation is that we can rewrite
    \begin{align*}
        \A_{G_\tau} = \mathbf{1}\mathbf{1}^\top - \Id - \A_{G\sbra{V \setminus (\tau \cup N\sbra{\tau})}}\mcom
    \end{align*}
    which is just the formula for the graph complement. Notice this is a principal submatrix of the entire complement $\overline{\A_G}$, so by the \nameref{fact:cauchyinterlacing} we have
    \begin{align*}
        \lambda_2(\A_{G_\tau}) \leq \lambda_2(\overline{\A_G})\mper
    \end{align*}
    So, to prove \eqref{eq:fixedgoal1} it is enough to show that 
    \begin{equation}
        \label{eq:fixgoal2}
        \lambda_2(\overline{\A_G}) \leq -\lambdamin(\A_G)-1\mper
    \end{equation}  
    By our earlier formula we have $\overline{\A_G} = \mathbf{1}\mathbf{1}^\top - \Id - \A_G$, so a quick application of \Cref{cor:cauchyinterlacing} yields
    \begin{align*}
        \lambda_2(\overline{\A_G}) \leq \lambda_2(\mathbf{1}\mathbf{1}^\top) + \lambdamax(-\Id - \A_G) = -\lambdamin(\A_G) - 1\mcom
    \end{align*}
    proving \eqref{eq:fixgoal2} as desired.
\end{proof}

\subsection{Neighborhood concentration in random regular graphs}

\begin{lemma}
    \label{lem:regularconc}
    Let $G=(V,E)$ be a random $\Delta$-regular graph. For any $\ell > 0$, there exists $\gamma = o_\Delta(1)$ such that with high probability over the randomness of $G$, every $\tau \subseteq V$ satisfies the following anti-expansion property: if $\abs{\tau} \leq \frac{\log \Delta}{(2+\gamma)\Delta} \abs{V}$ then
    \begin{align*}
            \abs{V \setminus (\tau \cup N\sbra{\tau})} \geq \frac{\ell\log^{3/2} \Delta}{\sqrt{\Delta}}\abs{V}\mper
        \end{align*}
\end{lemma}

The proof of this lemma is similar in spirit to \Cref{lem:bipartiteconc}, but we repeat it here for completeness.

\begin{proof}

   It suffices to prove this statement for all sets $\tau$ of size $\frac{\log \Delta}{(2+\gamma) \Delta} \abs{V}$, since the removal of vertices only causes the neighborhood shrink. 
    
    First, by \Cref{lem:singleregular} 
    \begin{equation}\E\sbra{\abs{V \setminus (\tau \cup N\sbra{\tau})}} \geq \frac{1 - o_\Delta(1)}{\Delta^{1/(2+\gamma)}}(\abs{V}-\abs{\tau}) = (\ell+2) \cdot \frac{\log^{3/2} \Delta}{\sqrt{\Delta}}\abs{V} \mcom
    \label{eq:regEYNtau}
    \end{equation}
    where $\gamma=\Theta(\frac{\log\log \Delta + \log \ell}{\log \Delta})$ is chosen such that the equality holds.
    
    Our goal is now to prove concentration of $|N\sbra{\tau}|$ around this expectation sufficient enough to union bound across all $\tau \subseteq X$ with $\abs{\tau} = \frac{\log \Delta}{(2+\gamma) \Delta} \abs{V}$.
    
    We start by fixing $\abs{\tau} = \frac{\log \Delta}{(2+\gamma) \Delta}\abs{V}$ and bounding the total number of sets in $\binom{V}{\abs{\tau}}$ by
    \begin{equation}\label{eq:regbinomest}
        \binom{\abs{V}}{\abs{\tau}} \leq \pbra{\frac{e\abs{V}}{\abs{\tau}}}^{\abs{\tau}} \leq \pbra{\frac{e(2+\gamma)\Delta}{\log \Delta}}^{\abs{\tau}} \leq \exp\pbra{\frac{\log^2 \Delta}{(2+\gamma)\Delta} \abs{V}}\mcom
    \end{equation}
    using standard binomial estimates and the assumption $\gamma < 1$. Using \eqref{eq:regEYNtau} and \eqref{eq:regbinomest} with the following lemma will be sufficient to finish.

    \begin{lemma}
        \label{lem:regmartingaletail}
        For any $t>0$, we have
        $$\Pr\sbra{ |N[\tau]| - \E\sbra{|N[\tau]|} > t} \leq \exp\pbra{ - \frac{1}{8} \cdot \frac{t^2 }{ \Delta\abs{\tau}}}\mper$$
    \end{lemma}

    To see why, we plug $t =  \frac{2\log^{3/2} \Delta}{\sqrt{\Delta}}\abs{V}$ into \Cref{lem:regmartingaletail} and apply a union bound over the entirety of $\binom{X}{\abs{\tau}}$ using our estimate \eqref{eq:regbinomest} to yield
    \begin{align*}
         \Pr\sbra{\forall \tau , \, \abs{N\sbra{\tau}} - \E\sbra{|N\sbra{\tau}|} \geq  \frac{2\log^{3/2} \Delta}{\sqrt{\Delta}} |V| } &\underset{\eqref{eq:regbinomest}}{\leq}  \exp\pbra{\frac{\log^2 \Delta}{(2+\gamma) \Delta} \abs{V}} \cdot \exp\pbra{ - \frac{\log^2 \Delta}{\Delta  }|V| } \\
         & \leq \exp\pbra{ - \frac{1}{2} \cdot \frac{\log^2 \Delta}{\Delta  }|V| }\mper
    \end{align*}

Conditioned on this event not happening, we have that for all such $\tau$, $|V \setminus (\tau \cup N[\tau])|$ satisfies
\begin{align*}
    \abs{V \setminus (\tau \cup N[\tau])} \underset{\eqref{eq:regEYNtau}}{\geq} \left((\ell+2) - 2\right) \frac{\log^{3/2} \Delta}{\sqrt{\Delta}}\abs{V} \geq \frac{\ell\log^{3/2} \Delta}{\sqrt{\Delta}} \abs{Y}\mcom
\end{align*}
as desired.
\end{proof}

\begin{proof}[Proof of \Cref{lem:regmartingaletail}]
    We define the following edge exposure martingale for the pairing model. Perhaps after renaming, we assume the $\frac{1}{\Delta}$-vertices in the cloud $C(\tau)$ are numbered $1,\dots,\Delta|\tau|$. For any such $i = 1, \dots, \Delta\abs{\tau}$ then, let $\pi_i \in V \times [\Delta]$  be the match for $i$. We then define
    \begin{align}\label{eq:regZidef}
        Z_i \coloneqq \E_{\pi_{i+1}, \dots, \pi_{\Delta\abs{\tau}}}\sbra{|N_\pi\sbra{\tau}| \mid \pi_{i}, \dots, \pi_1}\mper
    \end{align}
    Observe that $Z_0 = \E \sbra{|N\sbra{\tau}|}$ while $Z_{\Delta\abs{\tau}} = \abs{N\sbra{\tau}}$, i.e.~it is the number of neighbors of $\tau$ in this realization of the pairing model. As such, it is enough to show that
    \begin{equation}
        \Pr\sbra{Z_{\Delta|\tau|} - Z_0 \geq t}\leq \exp\pbra{ - \frac{1}{8} \cdot \frac{t^2 }{ \Delta\abs{\tau}}}\mper
    \end{equation}
    To do this, we use the Azuma-Hoeffding inequality.
    \begin{lemma}[Azuma-Hoeffding inequality]
        Let $\{\Delta_i \coloneqq Z_{i} - Z_{i-1}\}_{i \in [m]}$ be a martingale difference sequence satisfying $\abs{\Delta_i} \leq 2$ almost surely for every $i \in [m]$. Then for all $t >0$ we have
        \begin{align*}
            \Pr\sbra{Z_m- Z_0 > t} \leq \exp\pbra{ - \frac{1}{8} \cdot \frac{t^2 }{ m}}\mper
        \end{align*}
    \end{lemma}
    All that remains then is to justify the bounded differences assumption $\abs{Z_i - Z_{i-1}} \leq 2$ for all $i \in [\Delta\abs{\tau}]$ in our martingale. To do so, fix exposure results $\{\pi_j\}_{j = 1}^{i-1}$. Note if $i = \pi_j$ for $j = 1, \dots, i-1$, which is possible in the random regular pairing model, then $Z_i = Z_{i-1}$ and we are done. Otherwise, suppose $(v,k) \in C(V) \setminus \{\pi_j\}_{j = 1}^{i-1}$ and define 
    \begin{align*}
        f((v,k)) \coloneqq \E \sbra{|N\sbra{\tau}| \mid \pi_i = (v,k)}\mper
    \end{align*}
    The quantity we want to bound is then exactly $\max_{(v,k)\in C(V) \setminus \{\pi_j\}_{j=1}^{i-1}} |f((v,k)) - \E\sbra{f}|$, but it will suffice to bound $\abs{f((v_1, k_1)) - f((v_2, k_2))} \leq 2$ for any pairs $(v_1, k_1)$, $(v_2, k_2) \in C(V) \setminus \{\pi_j\}_{j = 1}^{i-1}$. 
    
    To do so, we first choose a (perfect) matching $M_1$ on $C(V)$ uniformly at random such that $1,\dots,i-1$ match to $\pi_1,\dots,\pi_{i-1}$ and $i$ matches to $(v_1,k_1)$. We make a new (perfect) matching $M_2$ as follows: say $(v_2,k_2)$ is matched (in $M_1$)  to $(v_3,k_3)$; we match $(v_3,k_3)$ to $(v_1,k_1)$ and we match $i$ to $(v_2,k_2)$, i.e., we switch the match of $(v_1,k_1)$ and $(v_2,k_2)$. It is not hard to see that this changes $|N[\tau]|$ by at most 2.
\end{proof}

\begin{proposition}
        \label{lem:singleregular}
        For any $\gamma < 1$, the pairing model on $V$, and $\tau \subseteq V$ with $\abs{\tau} = \frac{\log \Delta}{(2+\gamma) \Delta} \abs{V}$ we have
        \begin{align*}
             \Pr\sbra{v \notin N\sbra{\tau}} \geq \frac{1-o_\Delta(1)}{\Delta^{1/(2+\gamma)}}, \quad \forall v \in V \setminus \tau\mper
        \end{align*}
    \end{proposition}

    \begin{proof}
        We consider the probability that all $\frac{1}{\Delta}$-vertices in $C(v)$ avoid matching into $\tau$ or $C(v)$ sequentially, which provides a lower bound on the probability of matching into $\tau$. The probability that $(v,i)$ avoids $\tau \cup C(v)$, conditioned on all previous copies avoiding both, is then at least $1 - \frac{ \abs{\tau}+1}{ \abs{V} - 1}$. Therefore, the probability can be lower bounded by
    \begin{align*}
        \Pr\sbra{v \notin N\sbra{\tau}} \geq \left(1-\frac{\abs{\tau}+1}{|V|-1}\right)^{\Delta} \underset{1-x\geq e^{-x-x^2/2}}{\geq} e^{-\frac{\Delta(\abs{\tau}+1)}{|V|-1} - \frac{\Delta(\abs{\tau}+1)^2}{2(|V|-1)^2}}  \geq \frac{1-o_\Delta(1)}{\Delta^{1/(2+\gamma)}}\mper
    \end{align*}
\end{proof}
\section{Slow Mixing in the One-Sided Slice}
\label{sec:slowmix}

In this section, we show slow mixing of the one-sided chain at fugacity $\lambda = \Theta\pbra{\frac{1}{\sqrt{\Delta}}}$ for a bipartite graph $G = (X, Y, E)$ at slice size $k = \frac{|X|}{\Delta^{c}}$ for all $c \in (1/2,1)$. We show this by finding a graph exhibiting a disjoint partition where both sides have large volume yet low conductance, allowing us to conclude slow mixing using e.g. Cheeger's inequality. First, recall the definition of conductance.

\begin{definition}[Conductance]
For a reversible Markov chain $\P$ defined on vertex set $\Omega$ with stationary distribution $\mu$, the conductance of a subset $S \subseteq \Omega$ is given by 

$$\phi(S) = \frac{\mu (S, \Omega \setminus S)}{\min \{\mu(S), \mu(\Omega \setminus S) \}} $$

where $\mu(S) = \sum_{v \in S} \mu(v)$ and $\mu(A, B) = \sum_{a \in A, b \in B} \mu(a,b)$ where $\mu(a,b) = \mu(a) \P(a,b)$. The conductance $\phi \coloneqq \phi(\P)$ of $\P$ is then $\min_{\varnothing \subsetneq S \subsetneq \Omega} \phi(S) $. 
\end{definition}

The following classical theorem, (see e.g. \cite[Thm 7.4]{LevinPW17}) lower bounds the mixing time by the conductance.

\begin{theorem}
    \label{thm:slowmixing}
For any reversible chain $\P$ we have
$$\taumix(\P, 1/4) \geq \frac{1}{4\phi(\P)}\mper$$
\end{theorem}

We give the following simple construction: take a disjoint union of two independent random $\Delta$-regular bipartite graphs $G = G_1 \sqcup G_2$ with $n \coloneqq |V(G_1)| = |V(G_2)|$. For any subgraph $C \subseteq G$ we write $X(C), Y(C)$ to denote the left side and right side of the vertices of $C$ respectively.
For a set $S\subseteq {\cal X}(k)$ we let $\partial S \coloneqq \{ \tau \in S \mid \exists \, \sigma \in \mathcal{X}(k) \setminus S, |\tau \cap \sigma| = k -1 \}$ be the vertex boundary of $S$. For simplicity, we let $\mu$ denote the unnormalized measure of $\mu_\lambda^{(k)}$. We further assume $k$ is even in the analysis.

\begin{lemma}
Consider the one-sided slice of the hardcore model $(\calX, \mu_\lambda^{(k)})$ for $G$ above at fugacity $\lambda = \Theta\left(\frac{1}{\sqrt{\Delta}} \right)$ and $k = \frac{|X|}{\Delta^{c}}$ for $ c \in (1/2, 1)$. We have that with high probability the conductance $\phi$ satisfies $\phi \leq \exp( - \Omega(n/ \sqrt{\Delta}))$. Thus, the down-up walk mixes in exponential time.
\end{lemma}

\begin{proof}
 We prove the first assertion; the second assertion follows  by \Cref{thm:slowmixing}. 
 Let $S \coloneqq \{ \tau \in \mathcal{X}(k) \mid  |\tau \cap X(G_1)| > {k} / 2\}$. In words, they are the $k$-subsets of $X$ that share more than $k/2$ vertices with $X(G_1)$. Now, by definition, the conductance of $S$ is given by
 \begin{align*}
    \phi(S) = \frac{\mu(S, \mathcal{X}(k) \setminus S)}{\min \{\mu(S), \mu(\mathcal{X}(k) \setminus S) \}} \leq \frac{\mu(\partial S)}{\min \{\mu(S), \mu(\mathcal{X}(k) \setminus S) \}}\mper
 \end{align*}
 where the inequality is because the only edges out of $S$ comes from its boundary vertices. We claim that with high probability over randomness of  $G_1, G_2$ we have 
\begin{align}
     \mu(S), \mu (\mathcal{X}(k) \setminus S) &\geq \binom{n}{k} (1 + \lambda)^{n} \label{eq:conddenom}\\
     \mu( \partial S) &\leq \binom{n}{k/2} \binom{n}{k/2} (1 + \lambda)^{n/\sqrt{\Delta}} \mper\label{eq:condnum}
 \end{align}
Then,
 \begin{align*}
    \phi(S) &  \leq \frac{\binom{n}{k/2} \binom{n}{k/2} (1 + \lambda)^{o_\Delta(n)}}{{\binom{n}{k} (1 + \lambda)^{n}}} \\
    &\leq \frac{\left(\frac{en}{k/2}\right)^{k/2} \left(\frac{en}{k/2} \right)^{k/2}}{ \left(\frac{n}{k}\right)^k} (1 + \lambda)^{ - \Omega(n)} \\
    &= \exp(\Theta(k) - \Omega(\lambda n))
    \underset{k = \frac{|X|}{\Delta^c}, c \in (1/2, 1), \lambda = \Theta(\frac{1}{\sqrt{\Delta}})}{\leq} \exp(-\Omega(\frac{n}{\sqrt{\Delta}}))\mper
 \end{align*}
 as desired.
 
 It remains to prove \eqref{eq:conddenom}, \eqref{eq:condnum}. To see \eqref{eq:conddenom}, note
 $$ \mu(S)\geq \sum_{\tau\in {\cal X}(k) : \tau\subseteq X(G_1)} \mu(\tau) \geq \sum_{\tau\in {\cal X}(k): \tau\subseteq X(G_1)} (1+\lambda)^n = \binom{n}{k}(1+\lambda)^n\mcom$$
 where the second inequality uses that all vertices of $Y(G_2)$ have no edges to $X(G_1)$ (and $\tau$ in particular). The second inequality can be shown similarly.

 Next we show \eqref{eq:condnum}. First, since $k/2 = \frac{n}{2 \Delta^c}$, setting the constants appropriately, the condition of \Cref{lem:bipartitelargeconc} applies to $G_1$ and $k/2$ (and similarly $G_2$ and $k/2$) and we have that with high probability over the draw of $G_1, G_2$, it holds that for all $A\subseteq X(G_1), B\subseteq X(G_2)$ of size $|A| = k/2, |B| = k/2$, $\abs{Y(G_1) \setminus N\sbra{A}},\abs{Y(G_2) \setminus N\sbra{B}} \leq \frac{n}{2 \sqrt{\Delta}}$.

 Now recall by definition of $S$ we have $\partial S = \{\tau \in \mathcal{X}(k) \mid |\tau \cap X(G_1)| = \frac{k}{2}, |\tau \cap X(G_2)| = \frac{k}{2}\}$. Thus for all sets $\tau \in \partial S$, we have that 
 $$|Y \setminus N\sbra{\tau}| = |Y(G_1)\setminus N[\tau\cap X(G_1)]| + |Y(G_2)\setminus N[\tau\cap X(G_2)]|\leq \frac{n}{\sqrt{\Delta}} \mper$$ 
 Therefore, $\mu(\tau) \leq (1 + \lambda)^{\frac{n}{\sqrt{\Delta}}}$. \eqref{eq:condnum} follows since $|\partial(S)|=\binom{n}{k/2}\cdot \binom{n}{k/2}$. \end{proof}

\end{document}